\documentclass[10pt,reqno]{amsart}
\usepackage{bbm}
\usepackage{mathrsfs}
\usepackage{amsfonts} 
\usepackage[dvipsnames,usenames]{color}
\usepackage{graphicx,latexsym,bm,amsmath,amssymb,verbatim,multicol,lscape}
\newtheorem{thm}{Theorem} [section]
\newtheorem{lem}{Lemma}[section]

\theoremstyle{definition}
\newtheorem{defn}{Definition}[section]
\theoremstyle{remark}

\numberwithin{equation}{section}

\allowdisplaybreaks

\begin{document}
\title[A characterization for an almost MDS code
to be a near MDS code]{A characterization for an
almost MDS code to be a near MDS code and a
proof of the Geng-Yang-Zhang-Zhou conjecture}
\author[S.Y. Qiang]{Shiyuan Qiang}
\address{Mathematical College, Sichuan University,
Chengdu 610064, P.R. China}
\email{syqiang23@163.com}
\author[H.K. Wei]{Huakai Wei}
\address{Mathematical College, Sichuan University,
Chengdu 610064, P.R. China}
\email{1810097@mail.nankai.edu.cn}
\author[S.F. Hong]{Shaofang Hong$^*$}
\address{Mathematical College, Sichuan University,
Chengdu 610064, P.R. China}
\email{sfhong@scu.edu.cn}
\begin{abstract}
Let $\mathbb{F}_q$ be the finite field of $q$ elements,
where $q=p^{m}$ with $p$ being a prime number and $m$ being
a positive integer. Let $\mathcal{C}_{(q, n, \delta, h)}$
be a class of BCH codes of length $n$ and designed $\delta$.
A linear code $\mathcal{C}$ is said to be maximum distance
separable (MDS) if the minimum distance $d=n-k+1$. If $d=n-k$,
then $\mathcal{C}$ is called an almost MDS (AMDS) code.
Moreover, if both of $\mathcal{C}$ and its dual code
$\mathcal{C}^{\bot}$ are AMDS, then $\mathcal{C}$ is called
a near MDS (NMDS) code. In [A class of almost MDS codes,
{\it Finite Fields Appl.} {\bf 79} (2022),
\#101996], Geng, Yang, Zhang and Zhou proved that the BCH
code $\mathcal{C}_{(q, q+1,3,4)}$ is an almost MDS code,
where $q=3^m$ and $m$ is an odd integer, and they also
showed that its parameters is $[q+1, q-3, 4]$.
Furthermore, they proposed a conjecture stating that the dual
code $\mathcal{C}^{\bot}_{(q, q+1, 3, 4)}$ is also an AMDS code
with parameters $[q+1, 4, q-3]$. In this paper, we first
present a characterization for the dual code of an almost MDS
code to be an almost MDS code. Then we use this result to show
that the Geng-Yang-Zhang-Zhou conjecture is true. Our result
together with the Geng-Yang-Zhang-Zhou theorem implies that
the BCH code $\mathcal{C}_{(q, q+1,3,4)}$ is a near MDS code.
\end{abstract}
\thanks{$^*$S.F. Hong is the corresponding author and was
supported partially by National Science Foundation of China
\# 12171332.}
\keywords{Cyclic code; Linear code; MDS code; BCH code;
AMDS code; NMDS code}
\subjclass[]{11T71, 94B15, 94B05, 94A60}
\maketitle

\section{Introduction}
Let $\mathbb{F}_q$ be the finite field of $q$ elements,
where $q=p^{m}$ with prime $p$ and positive integer $m$,
$\mathbb{F}_q^{n}$ is an n-dimensional vector space
over $\mathbb{F}_q$. Let $\mathbb{F}_q^{*}
:=\mathbb{F}_q\backslash\{0\}$ be the set of all
non-zero elements of $\mathbb{F}_q$. Then $\mathbb {F}_q^*$
forms a group under the multiplicative operation.
For a non-empty set $\mathcal{C}\subseteq \mathbb{F}_q^{n}$,
if $\mathcal{C}$ is a $k$-dimensional subspace of $\mathbb{F}_q^{n}$,
then $\mathcal{C}$ is called an $[n, k]$ {\it linear code} over
$\mathbb{F}_q$, and $n$ and $k$ are called the {\it length} and
{\it dimension} of the code $\mathcal{C}$, respectively.
And we call the vector {\bf c}  in $\mathcal{C}$ {\it codeword}.
Then the linear code $\mathcal{C}$ holds $q^k$ codewords.
The {\it dual code}, denoted by $\mathcal{C}^{\bot}$, of an
$[n, k]$ linear code $\mathcal{C}$ over $\mathbb{F}_q$ is
defined by $\mathcal{C}^{\bot}:=\{{\bf c}^{\bot}\in
\mathbb{F}_q^{n}\mid\langle {\bf c}, {\bf c}^{\bot}\rangle=0,
{\bf c}\in \mathcal{C}\}$ with $\langle {\bf c}, {\bf c}^{\bot}\rangle$
denoting the Euclidean inner product of ${\bf c}$ and
${\bf c}^{\bot}$. For any integer $n\ge 1$, we set
$\langle n\rangle:=\{1, 2, \cdots, n\}$ throughout this paper.
As usual, for any finite set $A$, $|A|$ stands for its cardinality.

In coding theory, the two most common ways to present a linear
code are either a generator matrix or a parity-check matrix.
A {\it generator matrix} for an $[n, k]$ code $\mathcal{C}$ is
any $k\times n$ matrix $G$ whose rows form a basis of $\mathcal{C}$.
An $(n-k)\times n$ matrix $H$ is called a {\it parity-check matrix}
of an $[n, k]$ code $\mathcal{C}$. Then one has
$\mathcal{C}=\{{\bf x}\in\mathbb{F}_q^{n}\mid H{\bf x}^{T}=0\}$.
An important invariant of a code is the minimal distance
between codewords. Let ${\bf a}=(a_{1}, \cdots, a_{n})$,
${\bf b}=(b_{1}, \cdots, b_{n})$ be two codewords in
$\mathcal{C}$. Then the {\it Hamming weight} $w({\bf a})$ of ${\bf a}$
is defined to be the number of nonzero coordinates of ${\bf a}$,
that is, $w({\bf a})=|\{i\in\langle n\rangle \mid a_{i}\neq0\}|$.
The {\it support set} ${\rm supp}({\bf a})$ of ${\bf a}\in\mathbb{F}_q^{n}$
is defined by ${\rm supp}({\bf a}):=\{i\in\langle n\rangle\mid a_{i}\neq0\}$.
The {\it Hamming distance} $d({\bf a}, {\bf b})$ between ${\bf a}$ and ${\bf b}$
is defined to be the number of coordinates in which ${\bf a}$ and
${\bf b}$ differ. In other words, $d({\bf a}, {\bf b}):=|\{i\in\langle
n\rangle \mid a_{i}\neq b_{i}\}|$. Then $d({\bf a}, {\bf b})=w({\bf a}-{\bf b})$.
The {\it minimum Hamming distance}, denoted by
$d=d(\mathcal{C})$, of a code $\mathcal{C}\subseteq
\mathbb{F}_q^{n}$ is defined to be the smallest Hamming
distance between distinct codewords. Obviously, $d$
is equal to the minimal value of the Hamming weights
of all nonzero codewords in $\mathcal{C}$. Namely, $d=d(\mathcal{C})
:=\min_{\substack{{\bf c}_{1}, {\bf c}_{2}\in \mathcal{C},
{\bf c}_{1}\neq {\bf c}_{2}}}\{d({\bf c}_{1}, {\bf c}_{2})\}$. In particular,
the minimum Hamming distance is important in determining
the error-correcting capability of $\mathcal{C}$.
If the minimum Hamming distance $d$ of an $[n, k]$ code
is known, then we refer this code to an $[n, k, d]$ code.

For any integer $i$ with $0\leq i\leq n$, let $A_{i}$
stand for the number of codewords in $\mathcal{C}$
with length $n$ and the Hamming weight $i$. Then the
{\it weight enumerator} of $\mathcal{C}$ is defined
to be the polynomial
$$1+A_{1}x+A_{2}x^{2}+\cdots+A_{n}x^{n}.$$
The sequence $(1, A_{1}, A_{2},\cdots, A_{n})$ is called
the {\it weight distribution} of $\mathcal{C}$. The code
$C$ is called a {\it $t$-weight code} if
$|\{ A_{i}\neq0\mid i\in\langle n\rangle\}|=t$.
As an important type of linear codes, cyclic codes were
first studied by E. Prange in 1957. Not only they have
good algebraic structures, but also their encoding and
decoding can be easily implemented using the linear shift
registers \cite{[CS-HEP05], [HP-CUP03], [MS-1977], [RMR-CUP06]}.
We say that an $[n, k, d]$ code $\mathcal{C}$ over $\mathbb{F}_q$
is {\it cyclic} if $(c_{0}, c_{1}, \cdots, c_{n-1})\in\mathcal{C}$
implies that $(c_{n-1}, c_{0}, c_1, \cdots, c_{n-2})\in\mathcal{C}$.

In the rest of the paper, we impose the restriction 
$\gcd(n, q)=1$ and let $(x^{n}-1)$ be the ideal
$\mathbb{F}_q[x]$ generated by $x^{n}-1\in\mathbb{F}_q[x]$.
Then all the elements of the residue class ring
$\mathbb{F}_q[x]/(x^n-1)$ can be represented by polynomials
of degree less than $n$. Clearly, there is an isomorphism
$\varphi$ from $\mathbb{F}_q^{n}$ as a vector space over
$\mathbb{F}_q$ to $\mathbb{F}_q[x]/(x^n-1)$, defined by
\begin{align*}
\varphi: \mathbb{F}_q^{n}&\rightarrow\mathbb{F}_q[x]/(x^n-1)\notag\\
(c_{0}, c_{1}, \cdots, c_{n-1})&\mapsto c_{0}+c_{1}x+\cdots+c_{n-1}x^{n-1}.
\end{align*}
For convenience, we denote the elements of $\mathbb{F}_q[x]/(x^n-1)$
either as polynomials of degree no more than $n-1$ modulo $x^{n}-1$
or as vectors over $\mathbb{F}_q$. So we can interpret $\mathcal{C}$
as a subset of $\mathbb{F}_q[x]/(x^n-1)$. A linear code $\mathcal{C}$ is
cyclic if and only if $\mathcal{C}$ is an ideal of $\mathbb{F}_q[x]/(x^n-1)$.
Note that every ideal of $\mathbb{F}_q[x]/(x^n-1)$ is principal,
in particular, every nonzero ideal $\mathcal{C}$ is generated by 
the monic polynomial $g(x)$ of smallest degree in this ideal.

Let $\mathcal{C}=(g(x))$ be a cyclic code. Then $g(x)$ is called the
{\it generator polynomial} of $\mathcal{C}$ and $h(x)=(x^n-1)/g(x)$
is called the {\it parity-check polynomial} of $\mathcal{C}$. Let
$m={\rm ord}_{n}(q)$ be the order of $q$ modulo $n$, and
let $\alpha$ be a generator of $\mathbb{F}^{\ast}_{q^{m}}$.
Put $\beta=\alpha^{\frac{q^{m}-1}{n}}$. Then $\beta$ is a primitive
$n$-th root of unity in $\mathbb{F}_{q^{m}}$. For any integer $i$
with $0\leq i\leq n-1$, let $m_{i}(x)$ denote the minimal
polynomial of $\beta^{i}$ over $\mathbb{F}_q$. Let $h$ and
$\delta$ be two integers with $0\leq h\leq n-1$ and
$2\leq \delta\leq n$, define
$$g_{(q, n,\delta,h)}(x):={\rm lcm}(m_{h}(x), m_{h+1}(x), \cdots, m_{h+\delta-2}(x)),$$
where ${\rm lcm}$ denotes the least common multiple of these
minimal polynomials over $\mathbb{F}_q$, and the addition
in the subscript $b+i$ of $m_{b+i}(x)$ means the integer
addition modulo $n$. Let $\mathcal{C}_{(q, n,\delta,h)}$
denote the cyclic code of length $n$ and designed distance
$\delta$ over $\mathbb{F}_q$ with generator polynomial
$g_{(q, n,\delta,h)}(x)$. Then $\mathcal{C}_{(q, n,\delta,h)}$
is called a {\it BCH code} of length $n$ and designed distance
$\delta$. The maximum designed distance of a BCH code is called the
{\it Bose distance}. When $h=1$, $\mathcal{C}_{(q, n,\delta,h)}$
is called a {\it narrow-sense BCH code}. When $n=q^{m}-1$,
$\mathcal{C}_{(q, n,\delta,h)}$ is referred to as a
{\it primitive BCH code}.  It is well known that an $[n, k, d]$
code $\mathcal{C}$ is called a {\it maximum distance separable
(MDS) code} if the minimum distance $d$ reaches the Singleton
bound: $d\leq n-k+1$, i.e., if $d=n-k+1$. The dual of an MDS code
is also an MDS code. If  $d=n-k$, then the code is called an
{\it almost MDS code}. For brevity, we write an {\it AMDS}
code for an almost MDS code. Hence AMDS codes have
parameters $[n, k, n-k]$. In general, the dual of an AMDS
code may not be an AMDS code. Therefore determining whether
the dual of an AMDS code is an AMDS code is an interesting and
important problem, see \cite{[DB-DCC96]}, \cite{[LH-FFA23]},
\cite{[AMM-FFA22]} and \cite{[WH-DM21]} for some progress
about this topic. An AMDS code $\mathcal{C}$ is called a
{\it near MDS (NMDS) code} if its dual code is also an AMDS
code. MDS codes are widely applied in various fields due
to their nice properties, see \cite{[HP-CUP03]} and
\cite{[LN-CUP97]}. So the study of MDS code, AMDS codes and
NMDS code has attracted a lot of attention and made much vital
progress, see, for instance, \cite{[BGP-IT15]},
\cite{[DDK-IT97]}, \cite{[SI-JG95]}, \cite{[DL-DM00]}, \cite{[G-AAECC15]},
\cite{[HLW-IT22]} and \cite{[JK-IT19} .

BCH codes are a significant class of linear codes that
have good algebraic structure, are easy to construct, and
can be encoded and decoded relatively easily. Although BCH
codes have been studied for decades, their parameters
are not easy to determine. In the past 10 years, many
mathematicians devoted to this subject and made much
vital progress, see \cite{[DFZ-FFA17]}, \cite{[DT-IT20]},
\cite{[GYZZ-FFA22]}, \cite{[GDL-IT22]}, \cite{[LDXG-IT17]},
\cite{[LWL-IT19]} and \cite{[LLGS-DM23]}. In 2020, Ding and
Tang \cite{[DT-IT20]} considered the narrow-sense BCH codes
$\mathcal{C}_{(q, q+1,3,1)}$ with $q=3^m$ and its dual code.
They showed that the BCH code $\mathcal{C}_{(q, q+1,3,1)}$
and the dual code of $\mathcal{C}_{(q, q+1,3,1)}$ are NMDS
code. In 2022, Geng, Yang, Zhang and Zhou \cite{[GYZZ-FFA22]}
studied a class of AMDS codes from the BCH codes
$\mathcal{C}_{(q, q+1, 3, 4)}$ with $q=3^m$ and
determined their parameters.
\begin{thm}\label{thm1.1} \cite{[GYZZ-FFA22]}
Let $q=3^m$ with $m\geq3$ being odd. Then the BCH
code $\mathcal{C}_{(q, q+1, 3, 4)}$ is an AMDS code
with parameters $[q+1, q-3, 4]$.
\end{thm}
\noindent They verified through examples that the
code $\mathcal{C}^{\bot}_{(q, q+1, 3, 4)}$ is also
an AMDS code, where $q=3^m$ for $m=3, 5, 7$. Based
on this observation, they conjectured that if
$q=3^m$ with $m\geq3$ being odd, then the dual
code $\mathcal{C}^{\bot}_{(q, q+1, 3, 4)}$ is
an AMDS code with parameters $[q+1, 4, q-3]$.

In the current paper, our main goal is to
investigate this conjecture. We will prove that
the Geng-Yang-Zhang-Zhou conjecture mentioned
above is true as the following theorem shows.
\begin{thm}\label{thm1.2}
Let $q=3^m$ with $m\geq3$ being odd. Then the dual code
of the BCH code $\mathcal{C}_{(q, q+1, 3, 4)}$ is an AMDS
code with parameters $[q+1, 4, q-3]$.
\end{thm}

Combining Theorem \ref{thm1.1} with Theorem \ref{thm1.2},
we can conclude that the BCH code $\mathcal{C}_{(q, q+1,3,4)}$
is a NMDS code. That is, we have the following result.
\begin{thm}\label{thm1.3}
Let $q=3^m$ with $m\geq3$ being odd. Then the BCH
code $\mathcal{C}_{(q, q+1, 3, 4)}$ is a NMDS code
with parameters $[q+1, q-3, 4]$.
\end{thm}

This paper is organized as follows. First of all,
we establish in Section 2 a necessary and sufficient
condition for the dual code of an almost MDS code
to be an almost MDS code. Then in Section 3, we study
the BCH code $\mathcal{C}_{(q, q+1, 3, 4)}$ with
parameters $[q+1, q-3, 4]$ and give the proof of
Theorem \ref{thm1.2}.

\section{A characterization for the dual code of an
almost MDS code to be an almost MDS code}

In this section, we present a characterization
for the dual code of an almost MDS code to be an
almost MDS code that is needed in proving Theorem
\ref{thm1.2}. We begin with a renowned result
due to MacWilliams.

\begin{lem} \label{lem2} \cite{[CS-HEP05], [HP-CUP03],
[MS-1977], [RMR-CUP06]} {\rm (MacWilliams identity)}
For a linear $[n, k]$ code $\mathcal{C}$ over $\mathbb{F}_q$,
let $A_{i}$ and $A^{\bot}_{i}$ be the number of codewords with
the Hamming weight $i$ in $\mathcal{C}$ and $\mathcal{C}^{\bot}$,
respectively. Then for any integer $r$ with $0\leq r\leq n$, we have
$$\frac{1}{q^{k}}\sum_{i=0}^{n-r}\binom{n-i}{r}A_{i}
=\frac{1}{q^{r}}\sum_{i=0}^{r}\binom{n-i}{n-r}A^{\bot}_{i}.$$
\end{lem}

\begin{lem} \label{lem12} \cite{[RS-IT64]}
{\rm (Singleton bound)} For a linear $[n, k, d]$ code
$\mathcal{C}$ over $\mathbb{F}_q$, we have
$$d\leq n-k+1.$$
\end{lem}

\begin{lem} \label{lem13} \cite{[RMR-CUP06]}
Let $\mathcal{C}$ be a linear $[n, k, d]$ code over
$\mathbb{F}_q$. Then $\mathcal{C}$ is an MDS code
if and only if the dual code $\mathcal{C}^{\bot}$
of $\mathcal{C}$ is an MDS code.
\end{lem}

\begin{lem} \label{lem4}
Let $\mathcal{C}$ be an AMDS code over $\mathbb{F}_q$
with parameters $[n, k, n-k]$. Then $\mathcal{C}$ is
a NMDS code if and only if the following identity is true:
\begin{align} \label{2.0}
\frac{1}{q-1}(kA_{n-k}+A_{n-k+1})=\binom{n}{n-k+1}.
\end{align}
\end{lem}

\begin{proof}
As $\mathcal{C}$ is an AMDS code with parameters $[n, k, n-k]$,
applying Lemma \ref{lem2} to the case $r=k-1$, one derives that
$$\frac{1}{q^{k}}\sum_{i=0}^{n-k+1}\binom{n-i}{k-1}A_{i}
=\frac{1}{q^{k-1}}\sum_{i=0}^{k-1}\binom{n-i}{n-k+1}A^{\bot}_{i}.$$
Furthermore, the minimum distance of $\mathcal{C}$ equals $n-k$,
and so $A_i=0$ for any integer $i$ with $1\le i\le n-k-1$. It then
follows from $A_{0}=1=A_{0}^{\bot}$ that
\begin{align}\label{2.1}\frac{1}{q}\big(\binom{n}{k-1}
+\binom{k}{k-1}A_{n-k}+A_{n-k+1}\big)=
\binom{n}{n-k+1}+\sum_{i=1}^{k-1}\binom{n-i}{n-k+1}A^{\bot}_{i}.
\end{align}

First of all, we show the necessity part. To do so, one
lets $\mathcal{C} $ be a NMDS code. Then $\mathcal{C}^{\bot}$
is an AMDS code with $d(\mathcal{C}^{\bot})=k$.
This yields that $A^{\bot}_{i}=0$ for any integer $i$
with $1\le i\le k-1$. Hence
\begin{align}\label{2.4}
\sum_{i=1}^{k-1}\binom{n-i}{n-k+1}A^{\bot}_{i}=0.
\end{align}
Then from (\ref{2.1}), we deduce that
\begin{align}\label{2.2}\frac{1}{q}
\big(\binom{n}{k-1}
+kA_{n-k}+A_{n-k+1}\big)=\binom{n}{n-k+1}.
\end{align}
Therefore the desired result (\ref{2.0}) follows
immediately. The necessity is proved.

Finally, we show the sufficiency part. Assume
that (\ref{2.0}) is true. Combining (\ref{2.1})
and (\ref{2.2}) gives us (\ref{2.4}).
So by (\ref{2.4}), one can deduce that
$A^{\bot}_{i}=0$ for any integer $i$ with
$1\le i\le k-1$. This implies that
\begin{align}\label{2.5}
d(\mathcal{C}^{\bot})\geq k.
\end{align}
Noticing that $\mathcal{C}$ is an AMDS code with
parameters $[n, k, n-k]$, so $\mathcal{C}^{\bot}$ is
a code with parameters $[n, n-k, d(\mathcal{C}^{\bot})]$.
Since $\mathcal{C}$ is not an MDS code, Lemma
\ref{lem13} tells us that $\mathcal{C}^{\bot}$
is also not an MDS code. Thus from Lemma \ref{lem12}, Singleton
bound gives us that
\begin{align}\label{2.7}
d(\mathcal{C}^{\bot})\leq n-(n-k)=k.
\end{align}
It then follows from (\ref{2.5}) and (\ref{2.7}), we have
$d(\mathcal{C}^{\bot})=k$. Namely, $\mathcal{C}^{\bot}$
is an AMDS code with parameters $[n, n-k, k]$.
Consequently, $\mathcal{C}$ is a NMDS code as desired.
This ends the proof of Lemma \ref{lem4}.
\end{proof}

\begin{lem} \label{lem14}
Let $\mathcal{C}$ be an AMDS code with parameters $[n, k, n-k]$.
Then the Hamming weight of each nonzero codeword ${\bf c}$ in
$\mathcal{C}$ with ${\rm supp}({\bf c})\subseteq \{i_{1},
\cdots,i _{n-k+1}\}\subseteq \langle n\rangle$
is equal to either $n-k$, or $n-k+1$.
\end{lem}
\begin{proof}
Let ${\bf c}$ be a nonzero codeword in $\mathcal{C}$ with
${\rm supp}({\bf c})\subseteq \{i_{1},\cdots,i _{n-k+1}\}\subseteq \langle n\rangle$.
On the one hand, since $\mathcal{C}$ is an AMDS code with parameters
$[n, k, n-k]$, one has $w({\bf c})\geq d(\mathcal{C})=n-k$.
On the other hand, we have $w({\bf c})=|{\rm supp}({\bf c})|\leq |\{i_{1},
\cdots, i _{n-k+1}\}|=n-k+1$. Hence $n-k\leq w({\bf c})\leq n-k+1$.
So the expected result follows immediately.

This completes the proof of Lemma \ref{lem14}.
\end{proof}

\begin{defn}\label{def1}
Let $\mathcal{C}$ be an AMDS code with parameters $[n, k, n-k]$
and let $r$ be an integer with $1\le r\le n$, Then associated with
any given subset ${\{i_{1}, \cdots, i_{r}\}}\subseteq\langle n\rangle$,
we define a subset, denoted by $\mathcal{C}(i_{1}, \cdots, i_{r})$
or $\mathcal{C}(\{i_{1}, \cdots, i_{r}\})$, of $\mathcal{C}$ by
$$\mathcal{C}(i_{1}, \cdots, i_{r})
:=\{{\bf c}\in \mathcal{C}\mid {\rm supp}({\bf c})
\subseteq{\{i_{1}, \cdots, i_{r}\}}\}.$$
\end{defn}

Obviously, for any subset ${\{i_{1}, \cdots, i_{r}\}}
\subseteq\langle n\rangle$, the set $\mathcal{C}(i_{1}, \cdots, i_{r})$
may be empty and forms a vector space under as the usual vector
addition and scalar multiplication on $\mathbb{F}_q$. Now we
introduce two subsets of $\mathcal{C}$ as follows:
$$\mathcal{C}_i(i_{1}, \cdots, i_{r})
:=\{{\bf c}\in \mathcal{C}(i_{1}, \cdots, i_r)\mid w({\bf c})=\delta_i\},$$
where $\delta_i=n-k-1+i$ for $i\in \{1, 2\}$.

We have the following result.
\begin{lem} \label{lem5}
Let $\mathcal{C}$ be an AMDS code with parameters $[n, k, n-k]$.
Then for any subset ${\{i_{1}, \cdots, i_{n-k+1}\}}\subseteq
\langle n\rangle$, we have the following disjoint union:
$$\mathcal{C}(i_{1}, \cdots, i_{n-k+1})\backslash\{{\bf 0}\}=
\mathcal{C}_1(i_{1}, \cdots, i_{n-k+1})
\cup\mathcal{C}_2(i_{1}, \cdots, i_{n-k+1}),$$
where ${\bf 0}$ stands for the zero codeword of $\mathcal{C}$.
\end{lem}

\begin{proof}
First of all, it is obvious that $\mathcal{C}_1(i_{1}, \cdots, i_{n-k+1})
\cap\mathcal{C}_2(i_{1}, \cdots, i_{n-k+1})=\emptyset$.

Now for any ${\bf 0}\neq{\bf c}\in\mathcal{C}
(i_{1}, \cdots, i_{n-k+1})$, by Lemma \ref{lem14},
one has either $w({\bf c})=n-k$ or $w({\bf c})=n-k+1$.
If $w({\bf c})=n-k$, then ${\bf c}\in\mathcal{C}_1(i_{1}, \cdots, i_{n-k+1})$.
If $w({\bf c})=n-k+1$, then ${\bf c}\in\mathcal{C}_2(i_{1}, \cdots, i_{n-k+1})$.
By the arbitrariness of ${\bf c}$, one then concludes that
$$
\mathcal{C}(i_{1}, \cdots, i_{n-k+1})\subseteq\mathcal{C}_1(i_{1},
\cdots, i_{n-k+1})\cup\mathcal{C}_2(i_{1}, \cdots, i_{n-k+1}).
$$
But we have $\mathcal{C}_1(i_{1}, \cdots, i_{n-k+1})
\cup\mathcal{C}_2(i_{1}, \cdots, i_{n-k+1})\subseteq
\mathcal{C}(i_{1}, \cdots, i_{n-k+1}).$
So the expected result follows immediately.

This completes the proof of Lemma \ref{lem5}.
\end{proof}

\begin{lem}\label{lem6}
Let $\mathcal{C}$ be an AMDS code over $\mathbb{F}_q$
with parameters $[n, k, n-k]$ and let $i$ be a positive
integer no more than $n$. Then each of the following is true:
\begin{enumerate}
\item If the $i$-th component of all codewords in
$\mathcal{C}$ is equal to zero, then
$|\mathcal{C}(\langle n\rangle\backslash\{i\})|=q^{k}$
and $\dim\mathcal{C}(\langle n\rangle\backslash\{i\})=k$.
\item If there exists a codeword in $\mathcal{C}$ with the $i$-th
component nonzero, then we have $|\mathcal{C}(\langle n\rangle\backslash\{i\})|=q^{k-1}$
and $\dim\mathcal{C}(\langle n\rangle\backslash\{i\})=k-1$.
\end{enumerate}
\end{lem}
\begin{proof}
For arbitrary $\alpha\in\mathbb{F}_q$,
we define $\mathcal{C}_{i}(\alpha)$ as a subset of $\mathcal{C}$ by
$$\mathcal{C}_{i}(\alpha):=\{{\bf c}=({\bf c}_1, \cdots, {\bf c}_n)\in
\mathcal{C}\mid {\bf c}_{i}=\alpha\}.$$
We claim that $\mathcal{C}(\langle n\rangle\backslash\{i\})=\mathcal{C}_i(0)$.
In fact, for any ${\bf c}\in \mathcal{C}(\langle n\rangle\backslash\{i\})$,
one has ${\bf c}_i=0$. This implies that ${\bf c}\in\mathcal{C}_{i}(0)$. Then
$\mathcal{C}(\langle n\rangle\backslash\{i\})\subseteq\mathcal{C}_{i}(0)$.
Conversely, for any ${\bf c}\in\mathcal{C}_{i}(0)$, we must have ${\bf c}_i=0$.
But by Definition \ref{def1}, we have
$\mathcal{C}(\langle n\rangle\backslash\{i\})=
\{{\bf c}\in\mathcal{C}\mid{\rm supp}({\bf c})\subseteq\langle n\rangle\backslash\{i\}\}$
which implies that ${\rm supp}({\bf c})\subseteq\langle n\rangle\backslash\{i\}$.
Hence ${\bf c}\in\mathcal{C}(\langle n\rangle\backslash\{i\})$. That is,
$\mathcal{C}_{i}(0)\subseteq\mathcal{C}(\langle n\rangle\backslash\{i\})$.
Thus we have $\mathcal{C}(\langle n\rangle\backslash\{i\})=\mathcal{C}_i(0)$
as claimed.

In addition, the subset $\mathcal{C}_{i}(0)$
forms a vector subspace of $\mathcal{C}$ under vector
addition and scalar multiplication. It then follows that
$\dim\mathcal{C}(\langle n\rangle\backslash\{i\})
=\dim\mathcal{C}_i(0)$.

(1). Let ${\bf c}_{i}=0$ for any ${\bf c}\in \mathcal{C}$.
Then $\mathcal{C}_i(0)=\mathcal{C}$. Combining the above claim,
one yields that $\mathcal{C}(\langle n\rangle\backslash\{i\})
=\mathcal{C}_i(0)=\mathcal{C}$. Thus
$\# \mathcal{C}(\langle n\rangle\backslash\{i\}
=|\mathcal{C}_{i}(0)|=|\mathcal{C}|=q^{k}$
and $\dim\mathcal{C}(\langle n\rangle\backslash\{i\})
=\dim\mathcal{C}_{i}(0)=\dim\mathcal{C}=k$ as required.
Part (1) is proved.

(2). Assume that there exists a codeword ${\bf c}\in\mathcal{C}$
such that ${\bf c}_{i}\in\mathbb{F}^{*}_{q}$. For such ${\bf c}_i$,
one writes ${\bf c}_{i}=a$. Namely, ${\bf c}\in \mathcal{C}_{i}(a)$.
For any $b\in\mathbb{F}^{*}_q$, we have
$ba^{-1}{\bf c}\in \mathcal{C}$ with $(ba^{-1}{\bf c})_{i}
=ba^{-1}{\bf c}_i=ba^{-1}a=b$. Hence $ba^{-1}{\bf c}\in\mathcal{C}_{i}(b)$.
Therefore if we define a map $\theta$ from $\mathcal{C}_{i}(a)$
to $\mathcal{C}_{i}(b)$ by
\begin{align*}
\theta: \mathcal{C}_{i}(a)&\rightarrow\mathcal{C}_{i}(b), \notag\\
{\bf c}&\mapsto ba^{-1}{\bf c},
\end{align*}
then $\theta$ is well defined. First of all, for any
${\bf c}\in\mathcal{C}_{i}(b)$, one has ${\bf c}_i=b$.
This implies that $(ab^{-1}{\bf c})_i=ab^{-1}{\bf c}_i=ab^{-1}b=a$.
Thus $ab^{-1}{\bf c}\in\mathcal{C}_i(a)$. Moreover, we have
$\theta(ab^{-1}{\bf c})=ba^{-1}\cdot ab^{-1}{\bf c}={\bf c}$. Hence
$\theta$ is surjective. Consequently, for arbitrary
${\bf c}, {\bf c}'\in\mathcal{C}_{i}(b)$ satisfying ${\bf c}={\bf c}'$, one
can find $\bar {\bf c}, \bar {\bf c}'\in\mathcal{C}_{i}(a)$ such
that $\theta(\bar {\bf c})=ba^{-1}\bar {\bf c}={\bf c}$ and
$\theta(\bar {\bf c}')=ba^{-1}\bar {\bf c}'={\bf c}'$.
Since ${\bf c}={\bf c}'$, we have $ba^{-1}\bar {\bf c}=ba^{-1}\bar {\bf c}'$
which infers that $\bar {\bf c}=\bar {\bf c}'$. Hence $\theta$
is injective, and so the map $\theta$ is bijective.
It follows that
\begin{align}\label{2.8}
|\mathcal{C}_{i}(b)|
=|\mathcal{C}_{i}(a)|, \forall b\in\mathbb{F}^{*}_q.
\end{align}

Consequently, we show that the following is true:
\begin{align}\label{2.9}
|\mathcal{C}_{i}(0)|=|\mathcal{C}_{i}(a)|.
\end{align}
This will be done in what follows. On the one hand,
for any ${\bf d}=({\bf d}_1, \cdots, {\bf d}_n)\in\mathcal{C}_{i}(a)$, one has
$({\bf c}-{\bf d})_{i}={\bf c}_i-{\bf d}_i=a-a=0$ that yields that
${\bf c}-{\bf d}\in\mathcal{C}_{i}(0)$. So
$$\{{\bf c}-{\bf d}\mid {\bf d}\in\mathcal{C}_{i}(a)\}\subseteq\mathcal{C}_{i}(0).$$
It then follows that
\begin{align}\label{2.10}
|\mathcal{C}_{i}(a)|
=|\{{\bf d}\mid {\bf d}\in\mathcal{C}_{i}(a)\}|
=|\{{\bf c}-{\bf d}\mid {\bf d}\in\mathcal{C}_{i}(a)\}|
\leq|\mathcal{C}_{i}(0)|.
\end{align}

On the other hand, for any ${\bf f}=({\bf f}_1, \cdots, {\bf f}_n)\in\mathcal{C}_{i}(0)$,
one has $({\bf f}+{\bf c})_i={\bf f}_i+{\bf c}_i=0+a=a$.
Thus ${\bf f}+{\bf c}\in\mathcal{C}_i(a)$ which implies that
$$\{{\bf f}+{\bf c}\mid {\bf f}\in\mathcal{C}_{i}(0)\}
\subseteq\mathcal{C}_{i}(a).$$
Therefore
\begin{align}\label{2.11}
|\mathcal{C}_{i}(0)|=
|\{{\bf f}\mid {\bf f}\in\mathcal{C}_{i}(0)\}|=
|\{{\bf f}+{\bf c}\mid {\bf f}\in\mathcal{C}_{i}(0)\}|
\leq|\mathcal{C}_{i}(a)|.
\end{align}
Then the truth of (2.8) follows immediately from (\ref{2.10}) and (\ref{2.11}).

Hence, for any two distinct elements $b$ and $c$ in $\mathbb{F}^{*}_q$,
combining (\ref{2.8}) and (\ref{2.9}), we derive that
$$|\mathcal{C}_{i}(b)|=|\mathcal{C}_{i}(c)|=|\mathcal{C}_{i}(0)|$$
and
$$\mathcal{C}_{i}(b)\cap\mathcal{C}_{i}(c)=\{{\bf c}\in\mathcal{C}\mid {\bf c}_i=b\}
\cap\{{\bf c}\in\mathcal{C}\mid {\bf c}_i=c\}=\emptyset.$$
Furthermore, since
$$\mathcal{C}=\mathcal{C}_{i}(0)\cup(\bigcup_{a\in
\mathbb{F}^{*}_q}\mathcal{C}_{i}(a)),$$
one deduces that $|\mathcal{C}|
=q|\mathcal{C}_{i}(0)|$. So $|\mathcal{C}_{i}(0)|
=\frac{1}{q}|\mathcal{C}|=q^{k-1}$. Finally,
by the claim, one can deduce that
$|\mathcal{C}(\langle n\rangle\backslash\{i\})|
=|\mathcal{C}_{i}(0)|=q^{k-1}$
and $\dim\mathcal{C}(\langle n\rangle\backslash\{i\})
=\dim\mathcal{C}_{i}(0)=k-1$ as expected. Part (2) is proved.

This finishes the proof of Lemma \ref{lem6}.
\end{proof}

\begin{lem} \label{lem7}
Let $\mathcal{C}$ be an AMDS code with parameters $[n, k, n-k]$.
For any subset $\{i_{1}, \cdots, i_{n-k+1}\}\subseteq\langle n\rangle$,
we have $\dim\mathcal{C}(i_{1}, \cdots, i_{n-k+1})\in\{1,2\}$.
\end{lem}

\begin{proof}
For any subset $\{i_{1}, \cdots, i_{n-k+1}\}
\subseteq\langle n\rangle$, we write
$$\{i_{n-k+2}, \cdots, i_{n}\}:=\langle n\rangle\setminus\{i_{1}, \cdots, i_{n-k+1}\}.$$
Then
$$\{i_{1}, \cdots, i_{n-k+1}\}=\langle n\rangle\setminus
\{i_{n-k+2}, \cdots, i_{n}\}.$$
By Lemma \ref{lem6}, we know that
$$\dim\mathcal{C}({i_{1}, \cdots,  i_{n-k+1}})=
\dim\mathcal{C}(\langle n\rangle\backslash \{i_{n-k+2}\}
\backslash\cdots\backslash \{i_{n}\})\geq k-(k-1)=1.$$

Suppose that there exists subset
$\{i_{1}, \cdots, i_{n-k+1}\}\subseteq\langle n\rangle$
such that
$$\dim\mathcal{C}({i_{1}, \cdots, i_{n-k+1}})>2.$$
Then by Lemma \ref{lem6}, one obtains that
\begin{align*}
&\dim\mathcal{C}(i_{1}, \cdots, i_{n-k-1})\\
=&\dim(\mathcal{C}(i_{1}, \cdots, i_{n-k+1})\backslash\{i_{n-k+1}\})\backslash\{i_{n-k}\})\\
\geq &\dim\mathcal{C}(i_{1}, \cdots, i_{n-k+1})-2>0.
\end{align*}
It follows that there is at least a nonzero codeword
${\bf c}\in\mathcal{C}(i_{1}, \cdots, i_{n-k-1})$,
we have ${\rm supp}({\bf c})\subseteq \{i_{1},
\cdots, i_{n-k-1}\}$. Hence $w({\bf c})=| {\rm supp}({\bf c})|
\le|\{i_{1}, \cdots, i_{n-k-1}\}|=n-k-1$.
But $w({\bf c})\ge d({\mathcal C})=n-k$. Hence we arrive at
$$
n-k-1\ge w({\bf c})\ge d({\mathcal C})=n-k
$$
which is impossible. So
$\dim\mathcal{C}({i_{1}, \cdots, i_{n-k+1}})\le 2$
for any subset $\{i_{1}, \cdots, i_{n-k+1}\}
\subseteq\langle n\rangle$.
It follows that for any subset $\{i_{1}, \cdots,
i_{n-k+1}\}\subseteq\langle n\rangle$, we have
$\dim\mathcal{C}(i_{1}, \cdots, i_{n-k+1})\in\{1,2\}$
as desired. Lemma \ref{lem7} is proved.
\end{proof}

\begin{lem}\label{lem15}
Let $\mathcal{C}$ be an AMDS code with parameters
$[n, k, n-k]$. Let $\delta_i=n-k-1+i$ for $i\in\{1,2\}$.
Let $K_i\subseteq\langle n\rangle$ be an arbitrary subset
with $\delta_i\leq |K_i|\leq n$. Then each of the following is true:

{\rm  (1).} We have the disjoint union:
$$\mathcal{C}_i(K_i)=\bigcup_{\begin{subarray}{c}
I\subseteq K_i\\
|I|=\delta_i
\end{subarray}}\mathcal{C}_i(I).$$

{\rm (2).} We have
$$|\mathcal{C}_i(K_i)|=\sum_{\begin{subarray}{c}
I\subseteq K_i\\
|I|=\delta_i
\end{subarray}}|\mathcal{C}_i(I)|.$$
\end{lem}
\begin{proof}
Since part (2) is a obvious corollary of part (1), we need just
to show part (1). In the following, we show part (1).

Let $i\in\{1,2\}$ be given. Pick an arbitrary subset $I\subseteq K_i$
with $|I|=\delta_i$ and $\mathcal{C}_i(I)\neq \emptyset$.
We pick a ${\bf c}\in \mathcal{C}_i(I)$ with
$w({\bf c})=\delta_i$. Then ${\bf c}\in\mathcal{C}_i(K_i)$
which implies that $\mathcal{C}_i(I)\subseteq
\mathcal{C}_i(K_i)$. It follows that
$$\bigcup_{\begin{subarray}{c}
I\subseteq K_i |I|=\delta_i
\end{subarray}}\mathcal{C}_i(I)\subseteq \mathcal{C}_i(K_i).$$
Conversely, for any ${\bf c}\in\mathcal{C}_i(K_i)$,
writing ${\rm supp}({\bf c})=I$, we have
$${\bf c}\in \mathcal{C}_i(I)\subseteq\bigcup_{\begin{subarray}{c}
I\subseteq K_i\\
|I|=\delta_i
\end{subarray}}\mathcal{C}_i(I).$$
It then follows that
$$\mathcal{C}_i(K_i)
\subseteq\bigcup_{\begin{subarray}{c}
I\subseteq K_i |I|=\delta_i
\end{subarray}}\mathcal{C}_i(I).$$
Thus
$$\mathcal{C}_i(K_i)=\bigcup_{\begin{subarray}{c}
I\subseteq K_i\\
|I|=\delta_i
\end{subarray}}\mathcal{C}_i(I).$$
It remains to show that this union is disjoint
that will be done in what follows.

We claim that for $i\in\{1,2\}$, if $I_{1}, I_{2}\subseteq K$ with
$I_{1}\neq I_{2}$ and $ {\rm supp}(I_{1})={\rm supp}(I_{2})=\delta_i$,
then $\mathcal{C}_{i}(I_{1})\cap\mathcal{C}_{i}(I_{2})=\emptyset $.
Picking a ${\bf c}\in\mathcal{C}_{i}(I_{1})
\cap\mathcal{C}_{i}(I_{2})$ gives that
${\rm supp}({\bf c})\subseteq I_{1}\cap I_{2}$.
Since $I_{1}\neq I_{2}$,
one deduces that $|{\rm supp}({\bf c})|\leq |I_{1}
\cap I_{2}|<|I_1|=\delta_i$.
This contradicts with the fact
$|{\rm supp}({\bf c})|=\delta_i$. Hence
the proofs of the claim and  part (1) are complete.
So Lemma \ref{lem15} is proved.
\end{proof}

We introduce a concept. We say that a codeword
${\bf c}\in\mathcal{C}$ is {\it monic} if, as a vector,
the first nonzero component is the multiplicative
identity $1_{\mathbb{F}_q}$ of the finite field $\mathbb{F}_q$.

\begin{lem}\label{lem8}
Let $\mathcal{C}$ be an AMDS code over $\mathbb{F}_q$
with parameters $[n, k, n-k]$. Then it holds that
\begin{align}\label{2.6}
kA_{n-k}+ A_{n-k+1}=(q-1)\binom{n}{n-k+1}
\end{align}
if and only if it holds that
\begin{align}\label{6}
\dim \mathcal{C}(i_{1}, \cdots, i_{n-k+1})=1,
\forall\{i_{1}, \cdots, i_{n-k+1}\}\subseteq\langle n\rangle.
\end{align}
\end{lem}
\begin{proof}
First, we show the sufficiency part. We assume that
(\ref{6}) is true. Now for $i\in\{1,2\}$,
we define three sets $F_i$, $C_i$ and $E_i(t)$ as follows:
$$F_i:=\{\{i_1, ..., i_{n-k+1}\}\subseteq \langle n\rangle \mid
\mathcal{C}_i(i_{1}, \cdots, i_{n-k+1})\ne\emptyset\},$$
$$C_1:=\{{\bf c} \in\mathcal{C} \mid w({\bf c})=n-k\},
C_2:=\{{\bf c} \in\mathcal{C} \mid w({\bf c})=n-k+1\},$$
$$E_i(t):=\{{\bf c}\in C_i\mid {\rm\ the \ first \ nonzero
\ component \ of \ codeword}\ {\bf c} {\rm\ is} \ t\}
\ {\rm for} \ t\in\mathbb{F}_{q}^{*}.$$
Then for $i\in\{1,2\}$, $E_i(t)\cap E_i(t')=\emptyset$ if
$t, t'\in\mathbb{F}_{q}^{*}$ and $t\ne t'$. Moreover,
$|C_1|=A_{n-k}$ and $|C_2|=A_{n-k+1}$. For brevity, we
let $E_i:=E_i(1_{\mathbb{F}_{q}})$ for $i=1, 2$.

{\sc Claim I.} We have $|E_i(t)|=|E_i| \ (i=1, 2)$
for any $t\in\mathbb{F}_{q}^{*}$.

In fact, for any codeword ${\bf c_{t}}\in E_i(t)$, let
$({\bf c_{t}})_{i_{1}}$ be the first nonzero component
of ${\bf c_{t}}$ with $1\le i_1\le n$. Then
$({\bf c_{t}})_{i_{1}}^{-1}{\bf c_{t}}\in E_i$, from which  one derives
$E_i(t)\subseteq E_i $.
It then follows that $|E_i(t)|\leq|E_i|$. Conversely, let
${\bf c}\in E_i$ be any given element. Since ${\bf c}$
is monic, one has $t{\bf c}\in E_i(t)$. Thus
$E_i\subseteq E_i(t)$, and so $|E_i(t)|\geq |E_i|$
from which one deduces that $|E_i(t)|=|E_i|$ as claimed.
So Claim I is proved.

For $i\in\{1,2\}$, since all $E_i(t), t\in\mathbb{F}_{q}^{*}$,
are disjoint, we have
$\bigcup_{t\in \mathbb{F}_{q}^{*}}E_i(t)=C_i$.
So by Claim I,
$$|C_i|=\sum_{t\in \mathbb{F}_{q}^{*}}|E_i(t)|=(q-1)|E_i|.$$
It follows that
\begin{align}\label{2.3}
|E_1|=\frac{1}{q-1}|C_1|=\frac{A_{n-k}}{q-1}
\end{align}
and
\begin{align}\label{2.19}
|E_2|=\frac{1}{q-1}|C_2|=\frac{A_{n-k+1}}{q-1}.
\end{align}

{\sc Claim II.} We have $|F_1|=k|E_1|$.

In order to prove Claim II, we define the map
\begin{align*}
\vartheta: & E_1 \rightarrow\tilde{M}_1:=
\{{\rm supp}({\bf c})\mid {\bf c}\in E_1\} \\
& {\bf c}\mapsto {\rm supp}({\bf c}).
\end{align*}

Clearly, it is well
defined and is surjective. Now we show it is injective.
In fact, let ${\rm supp}({\bf c})={\rm supp}({\bf c}')$,
where ${\bf c}, {\bf c}'\in E_1$. Then $w({\bf c})=
w({\bf c}')=n-k$ and ${\bf c}$ and ${\bf c}'$ are monic.
One may write ${\rm supp}({\bf c})
={\rm supp}({\bf c}')=\{j_1, ..., j_{n-k}\}
(\subseteq\langle n\rangle)$ with $j_1<...<j_{n-k}$.
Since ${\bf c}$ and ${\bf c}'$ are monic, we have
${\rm supp}({\bf c}-{\bf c}')\subseteq\{j_2, ..., j_{n-k}\}$
and so $w({\bf c}-{\bf c}')=|{\rm supp}({\bf c}-{\bf c}')|
\le n-k-1$. However, ${\bf c}-{\bf c}'\in {\mathcal C}$ and
$\mathcal{C}$ is an $[n, k, n-k]$ AMDS code. Hence we
must have ${\bf c}-{\bf c}'={\bf 0}$, the zero codeword of
${\mathcal C}$. In other words, ${\bf c}={\bf c}'$. Thus
$\vartheta$ is injective, and so is a bijection. Therefore
$|E_1|=|\tilde{M}_1|$.

Now we define a set $M_1$ by
$$M_1:=\{{\rm supp}({\bf c})\cup\{j_{{\bf c}}\} \mid {\bf c}\in E_1,
j_{{\bf c}}\in\langle n\rangle\backslash{\rm supp}({\bf c})\}.$$
Then
$M_1=\{{\rm supp}({\bf c})\cup\{j_{{\bf c}}\} \mid
{\bf c}\in\mathcal{C}, w({\bf c})=n-k, {\bf c} \ {\rm is \ monic},
\ j_{{\bf c}}\in\langle n\rangle\backslash{\rm supp}({\bf c})\}.$
Now let $N:={\rm supp}({\bf c})\cup \{j_{\bf c}\}
={\rm supp}({\bf c}')\cup \{j_{\bf c}'\}$, where
${\bf c}, {\bf c}'\in E_1$, $j_{\bf c}\in\langle n\rangle
\backslash{\rm supp}({\bf c})$,
$j_{{\bf c}'}\in\langle n\rangle\backslash{\rm supp}({\bf c}')$.
Then $\# N=n-k+1$, ${\rm supp}({\bf c})\subseteq N$
and ${\rm supp}({\bf c}')\subseteq N$.
It then follows that
${\bf c}, {\bf c}'\in\mathcal{C}_1(N)\subseteq\mathcal{C}(N)$.
By the condition (\ref{6}), one has $\dim \mathcal{C}(N)=1$.
This implies that ${\bf c}=\ell{\bf c}'$ with
$\ell\in \mathbb{F}^{*}_q$. Since ${\bf c} $ and ${\bf c}'$
are monic, one has $\ell=1$ which implies that ${\bf c}={\bf c}'$.
Thus $j_{\bf c}=j_{{\bf c}'}$. Hence any two elements in
the set $M_1$ are distinct.

From this and noticing that
${\rm supp}({\bf c})=n-k$ for any ${\bf c}\in E_1$,
we then derive that
\begin{align}\label{2.14}
|M_1|=|\tilde{M}_1|\cdot|(\langle n\rangle\backslash{\rm supp}
({\bf c}))|=|E_1|\cdot|(\langle n\rangle\backslash{\rm supp}
({\bf c}))|=|E_1|\cdot (n-(n-k))=k|E_1|.
\end{align}

Now we prove that $F_1=M_1$. On the one hand, for any
${\rm supp}({\bf c})\cup\{j_{{\bf c}}\}\in M_1$,
we have ${\bf c} \in E_1$ and $j_{{\bf c}}\in\langle n\rangle
\backslash{\rm supp}({\bf c})$. Thus
$|{\rm supp}({\bf c})|=n-k$. But ${\rm supp}({\bf c})
\subseteq{\rm supp}({\bf c})\cup\{j_{{\bf c}}\}$.
Hence ${\bf c}\in{\mathcal C}_1({\rm supp}({\bf c})
\cup\{j_{{\bf c}}\})$, namely, ${\mathcal C}_1({\rm supp}
({\bf c})\cup\{j_{{\bf c}}\})\ne\emptyset$. It the follows
from $| {\rm supp}({\bf c})\cup\{j_{{\bf c}}\}|=n-k+1$ that
${\rm supp}({\bf c})\cup\{j_{{\bf c}}\}\in F_1$.
By the arbitrariness of ${\rm supp}({\bf c})\cup\{j_{{\bf c}}\}\in M_1$,
one then concludes that $M_1\subseteq F_1$.
It remains to show that $F_1\subseteq M_1$. Actually,
for any $\{i_{1}, \cdots, i_{n-k+1}\}\in F_1$,
we have $\mathcal{C}_1(i_{1}, \cdots, i_{n-k+1})\ne\emptyset$.
Pick an ${\bf e}\in\mathcal{C}_1(i_{1}, \cdots, i_{n-k+1})$.
Then $w({\bf e})=n-k$ and ${\rm supp}({\bf e})\subseteq\{i_{1},
\cdots, i_{n-k+1}\}$. Consider the following two cases.

{\it Case 1.} ${\bf e}_{i_{1}}=0$. Then
${\bf e}_{i_{2}}\neq 0$ and ${\rm supp}({\bf e}_{i_{2}}
^{-1}{\bf e})={\rm supp}({\bf e})=\{i_{2}, \cdots, i_{n-k+1}\}$.
So $w({\bf e}_{i_{2}}^{-1}{\bf e})=|{\rm supp}
({\bf e}_{i_{2}}^{-1} {\bf e})|=|{\rm supp}({\bf e})|=n-k$.
It implies that
${\bf e}^{-1}_{i_{2}} {\bf e}\in E_1$. At this moment, we have
$$\{i_{1}, i_2, \cdots, i_{n-k+1}\}={\rm supp}({\bf e})\cup\{i_1\}
={\rm supp}({\bf e}_{i_{2}}^{-1}{\bf e})\cup\{i_1\}\in M_1.$$
So we have $F_1\subseteq M_1$ as desired.

{\it Case 2.} ${\bf e}_{i_{1}}\neq 0$. Likewise, we have
${\bf e}^{-1}_{i_{1}}{\bf e}\in E_1$, which infers that
${\rm supp}({\bf e}^{-1}_{i_{1}}{\bf e})={\rm supp}({\bf e})
\subseteq\{i_{1}, \cdots, i_{n-k+1}\}$.
Hence we arrive at
\begin{align*}
\{i_{1}, \cdots, i_{n-k+1}\}=& {\rm supp}({\bf e})
\cup(\{i_{1}, \cdots, i_{n-k+1}\}\backslash{\rm supp}({\bf e}))\\
=& {\rm supp}({\bf e}^{-1}_{i_{1}}{\bf e})
\cup(\{i_{1}, \cdots, i_{n-k+1}\}\backslash{\rm supp}
({\bf e}^{-1}_{i_{1}}{\bf e}))\in M_1.
\end{align*}
By the arbitrariness of $\{i_{1}, \cdots, i_{n-k+1}\}\in F_1$,
one has $F_1\subseteq M_1$ as required. One then concludes
that $F_1=M_1$.

Finally, from (\ref{2.14}) one derives that
$|F_1|=|M_1|=k|E_1|$ as claimed. Claim II is proved.

Consequently, combining  (\ref{2.3}) and Claim II,
we derive that
\begin{align}\label{4''}
|F_1|= k|E_1|=\frac{kA_{n-k}}{q-1}.
\end{align}

{\sc Claim III.} We have $|F_2|=|E_2|$.

We define a set $M_2$ by $M_2:=\{{\rm supp}
({\bf c})\mid {\bf c}\in E_2\}$. Then
$M_2=\{{\rm supp}({\bf c})\mid
{\bf c}\in\mathcal{C}, w({\bf c})=n-k+1,
{\bf c} \ {\rm is \ monic}\}$.
Let $N':={\rm supp}({\bf c})={\rm supp}({\bf c}')$,
where ${\bf c}, {\bf c}'\in E_2$.
Then $|N'|=n-k+1$. It then follows that
${\bf c}, {\bf c}'\in\mathcal{C}_2(N')\subseteq\mathcal{C}(N')$.
By the condition (\ref{6}), one has $\dim \mathcal{C}(N')=1$.
This implies that ${\bf c}=\ell{\bf c}'$ with
$\ell\in \mathbb{F}^{*}_q$. Since ${\bf c} $ and
${\bf c}'$ are monic, one has $\ell=1$
which implies that ${\bf c}={\bf c}'$. Thus any two
elements in the set $M_2$ are distinct. It then
follows that
\begin{align}\label{2.15}
|M_2|=|E_2|.
\end{align}

Now we prove that $F_2=M_2$. On the one hand,
take arbitrarily ${\rm supp}({\bf e})\in M_2$.
Then we have ${\bf e} \in E_2$. It is obvious that
$|{\rm supp}({\bf e})|=n-k+1$. Then ${\rm supp}
({\bf e})\in F_2$. By the arbitrariness of
${\rm supp}({\bf e})\in M_2$,
one concludes that $M_2\subseteq F_2$.
On the other hand, for any $\{i_{1}, \cdots,
i_{n-k+1}\}\in F_2$, we have
$\mathcal{C}_2(i_{1}, \cdots, i_{n-k+1})\ne\emptyset$.
Then picking an ${\bf e'}\in\mathcal{C}_2(i_{1},
\cdots, i_{n-k+1})$ gives us that
${\rm supp}({\bf e'})=\{i_{1}, \cdots, i_{n-k+1}\}$
and $w({\bf e'})=n-k+1$. It follows that
$({\bf e}')^{-1}_{i_{1}}{\bf e}'\in E_2$. Therefore
$$\{i_{1}, \cdots, i_{n-k+1}\}={\rm supp}({\bf e}')
={\rm supp}(({\bf e}')^{-1}_{i_{1}}{\bf e}')\in M_2.$$
By the arbitrariness of $\{i_{1}, \cdots, i_{n-k+1}\}\in F_2$,
one has $F_2\subseteq M_2$, from which one derives that
$F_2=M_2$ as expected.

We can now deduce from (\ref{2.15}) that
$|F_2|=|M_2|=|E_2|$ as asserted. So Claim III
is proved.

Combining  (\ref{2.19}) and Claim III, we can easily deduce that
\begin{align}\label{4}
|F_2|= |E_2|= \frac{A_{n-k+1}}{q-1}.
\end{align}

Now from (\ref{4''}) and (\ref{4}), one yields that
\begin{align}\label{2.20}
|F_1|+|F_2|=\frac{1}{q-1}( kA_{n-k}+ A_{n-k+1}).
\end{align}

In what follows, we show that $F_1\cap F_2=\emptyset$.
We assume that $F_1\cap F_2\ne \emptyset$. Then we can
choose a set $\{i_{1}, \cdots, i_{n-k+1}\}\in F_1\cap F_2$.
So $\{i_{1}, \cdots, i_{n-k+1}\}\in F_1$ and
$\{i_{1}, \cdots, i_{n-k+1}\}\in F_2$. This infers
that there exist ${\bf f}, {\bf f'}\in\mathcal{C}$
such that ${\bf f}\in\mathcal{C}_1(i_{1}, \cdots, i_{n-k+1})$
and ${\bf f'}\in\mathcal{C}_2(i_{1}, \cdots, i_{n-k+1})$.
One can deduce that $w({\bf f})=n-k$ and $w({\bf f}')=n-k+1$.
Moreover, we have ${\bf f}, {\bf f'}\in\mathcal{C}(i_{1},
\cdots, i_{n-k+1})$. However, the condition (\ref{6})
tells us that $\dim \mathcal{C}(i_{1}, \cdots, i_{n-k+1})=1$.
Hence we must have ${\bf f}=\ell{\bf f}'$ with
$\ell\in \mathbb{F}^{*}_q$. This implies that
$w({\bf f})=w({\bf f}')$ which is impossible since
$w({\bf f})=n-k$ and $w({\bf f}')=n-k+1$.
Hence we have $F_1\cap F_2=\emptyset$ as expected.

Now for any subset$\{i_{1}, \cdots, i_{n-k+1}\}
\subseteq \langle n\rangle$, Lemma \ref{lem7} tells us
$\dim\mathcal{C}(i_{1}, ..., i_{n-k+1})\in\{1,2\}$ which
implies that $\mathcal{C}(i_{1}, \cdots, i_{n-k+1})
\ne\emptyset$. Finally, applying the inclusion-exclusion
principle (see, for example, \cite{[Ap]}) and Lemma
\ref{lem5}, we arrive at
\begin{align}\label{2.12}
|F_1|+|F_2|=&|(F_1\cup F_2)|+|(F_1\cap F_2)| \notag\\
=&|(F_1\cup F_2)| \notag\\
=&|\{\{i_{1}, \cdots, i_{n-k+1}\}\subseteq\langle n\rangle
\mid \mathcal{C}_1(i_{1}, \cdots, i_{n-k+1})
\cup\mathcal{C}_2(i_{1}, \cdots, i_{n-k+1})\ne\emptyset\}| \notag\\
=&|\{\{i_{1}, \cdots, i_{n-k+1}\}\subseteq\langle n\rangle
\mid \mathcal{C}(i_{1}, \cdots, i_{n-k+1})
\backslash\{{\bf 0}\}\ne\emptyset\}|\notag\\
=& |\{\{i_{1}, \cdots, i_{n-k+1}\} \mid
\{i_{1}, \cdots, i_{n-k+1}\}\subseteq\langle n\rangle\}|\notag\\
=&\binom{n}{n-k+1}.
\end{align}
Therefore, by the (\ref{2.20}) and (\ref{2.12}), the desired
result (\ref{2.6}) follows immediately. So the sufficiency part
is proved.

Now we are in the position to give the proof of the
necessity part. Suppose that (\ref{6}) is true.
Letting $K_i=\langle n\rangle$ $(i=1, 2)$ in Lemma
\ref{lem15} tells us that
\begin{align}\label{2.21}
kA_{n-k}= &k|C_1|\notag\\
=& k|\mathcal{C}_{1}(\langle n\rangle)|\notag\\
=& k\sum_{\begin{subarray}{c}
I\subseteq\langle n\rangle\\
|I|=\delta_1=n-k
\end{subarray}} |\mathcal{C}_{1}(I)|\notag\\
=& \sum_{\begin{subarray}{c}
I\subseteq\langle n\rangle\\
|I|=n-k
\end{subarray}}|\mathcal{C}_{1}(I)|\sum_{j\in
\langle n\rangle\backslash I}1\notag\\
=&\sum_{\begin{subarray}{c}
I\subseteq\langle n\rangle\\
|I|=n-k
\end{subarray}}\sum_{j\in\langle n\rangle\backslash I}
|\mathcal{C}_{1}(I)|\notag\\
=& \sum_{\begin{subarray}{c}
J\subseteq\langle n\rangle\\
|J|=n-k+1
\end{subarray}}\sum_{j\in J}
|\mathcal{C}_{1}(J\backslash\{j\})|\notag\\
=& \sum_{\begin{subarray}{c}
J\subseteq\langle n\rangle\\
|J|=n-k+1
\end{subarray}}\sum_{\begin{subarray}{c}
I\subseteq J\\
|I|=n-k
\end{subarray}}|\mathcal{C}_{1}(I)|\notag\\
=& \sum_{\begin{subarray}{c}
J\subseteq\langle n\rangle\\
|J|=n-k+1
\end{subarray}}|\mathcal{C}_{1}(J)|
\end{align}
and
\begin{align}\label{2.22}
A_{n-k+1}= &|C_2|\notag\\
=& |\mathcal{C}_{2}(\langle n\rangle)|\notag\\
=& \sum_{\begin{subarray}{c}
J\subseteq\langle n\rangle\\
|J|=\delta_2=n-k+1
\end{subarray}} |\mathcal{C}_{2}(J)|.
\end{align}

It then follows from (\ref{2.21}) and (\ref{2.22}) that
\begin{align}\label{2.23}
 kA_{n-k}+A_{n-k+1}
=& \sum_{\begin{subarray}{c}
J\subseteq\langle n\rangle\\
|J|=n-k+1
\end{subarray}}|\mathcal{C}_{1}(J)|
+\sum_{\begin{subarray}{c}
J\subseteq\langle n\rangle\\
|J|=n-k+1
\end{subarray}}|\mathcal{C}_{2}(J)|\notag\\
=&\sum_{\begin{subarray}{c}
J\subseteq\langle n\rangle\\
|J|=n-k+1
\end{subarray}}(|\mathcal{C}_{1}(J)|+|\mathcal{C}_{2}(J)|)\notag\\
=&\sum_{\begin{subarray}{c}
J\subseteq\langle n\rangle\\
|J|=n-k+1
\end{subarray}}(|\mathcal{C}_{1}(J)
\cup\mathcal{C}_{2}(J)|+|\mathcal{C}_{1}(J)
\cap|\mathcal{C}_{2}(J)|)\notag\\
=&\sum_{\begin{subarray}{c}
J\subseteq\langle n\rangle\\
|J|=n-k+1
\end{subarray}}|\mathcal{C}_{1}(J)
\cup\mathcal{C}_{2}(J)| \ {\rm (since} \ \mathcal{C}_{1}(J)\cap|\mathcal{C}_{2}(J)=\emptyset)\notag\\
=&\sum_{\begin{subarray}{c}
J\subseteq\langle n\rangle\\
|J|=n-k+1
\end{subarray}}|\mathcal{C}(J)
\backslash\{{\bf 0}\}|\ {\rm (by \ Lemma \ref{lem5})}\notag\\
=&\sum_{\begin{subarray}{c}
J\subseteq\langle n\rangle\\
|J|=n-k+1
\end{subarray}}(q^{\dim\mathcal{C}(J)}-1).
\end{align}
But Lemma \ref{lem7} gives us that for any subset
$J\subseteq\langle n\rangle$ with $|J|=n-k+1$,
one has $\dim\mathcal{C}(J)\in\{1, 2\}$, and hence
$q^{\dim\mathcal{C}(J)}-1=q-1$, or $q^2-1$.
It then follows from (\ref{2.23}) and (\ref{2.6}) that

\begin{align*}
kA_{n-k}+A_{n-k+1}=&\sum_{\begin{subarray}{c}
J\subseteq\langle n\rangle\\
|J|=n-k+1
\end{subarray}}(q^{\dim\mathcal{C}(J)}-1)\\
\geq &\sum_{\begin{subarray}{c}
J\subseteq\langle n\rangle\\
|J|=n-k+1
\end{subarray}}(q-1)\\
=&(q-1)\sum_{\begin{subarray}{c}
J\subseteq\langle n\rangle\\
|J|=n-k+1
\end{subarray}}1\\
=&(q-1)\binom{n}{n-k+1}\\
=& kA_{n-k}+A_{n-k+1}.
\end{align*}
This enables $\dim\mathcal{C}(J)=1$ for any subset
$J\subseteq\langle n\rangle$ with $|J|=n-k+1$. That
is, (\ref{6}) holds as one desires. So the necessity
part is proved.

The proof of Lemma \ref{lem8} is complete.
\end{proof}

Finally, we can give a characterization for
the dual code of an AMDS code to be an AMDS
code. In other words, we obtain a characterization
for an AMDS code to be a NMDS code. This is
the main result of this section.

\begin{thm} \label{thm2.1}
Let $\mathcal{C}$ be an AMDS code over $\mathbb{F}_q$
with parameters $[n, k, n-k]$. Then the dual code of
$\mathcal{C}$ is an AMDS code if and only if
$\dim\mathcal{C}(I)=1$ holds for any subset
$I\subseteq\langle n\rangle$ with $|I|=n-k+1$.
\end{thm}
\begin{proof}
This follows immediately from Lemmas \ref{lem4} and \ref{lem8}.
\end{proof}

\section{On the BCH code $\mathcal{C}_{(q, q+1, 3, 4)}$
with parameters $[q+1, q-3, 4]$ and proof of Theorem \ref{thm1.2}}

Now we turn our attention to the BCH code
$\mathcal{C}_{(q, q+1, 3, 4)}$. We will
explore its properties. Throughout this
section, let $U_{q+1}$ denote the cyclic
group consisting of all $(q+1)$-th roots
of unity in $\mathbb{F}_{q^{2}}$, which
is a subgroup of the multiplicative
group of $\mathbb{F}_{q^{2}}$.

\begin{lem}\label{lem3.1}
Let $q=p^m$ with $p$ being a prime number and $m\ge 1$
being an integer and let $\mathcal{C}$ be the BCH code
$\mathcal{C}_{(q, q+1, 3, 4)}$. Then
$\dim \mathcal{C}(i_{1}, \cdots, i_{n-k+1=5})=1$ holds for
any subset $\{i_{1}, \cdots, i_{5}\}\subseteq\langle n\rangle$
if and only if for arbitrary two distinct elements
$x, y\in U_{q+1}\backslash\{1:=1_{\mathbb{F}_{q^{2}}}\}
\subseteq \mathbb{F}_{q^{2}}$, there exists at most one element
$z\in U_{q+1}\backslash\{x, y, 1\}\subseteq{\mathbb{F}_{q^{2}}}$
such that the following system of homogeneous equations
\begin{align}\label{3.1}
\left(\begin{array}{cccc}
x^{4} & y^{4} & z^{4} & 1 \\
x^{5} & y^{5} & z^{5} & 1 \\
\end{array}\right){\bf X}={\bf 0}
\end{align}
has a nonzero solution ${\bf X}$ in ${\mathbb{F}}^{4}_{q}$.
\end{lem}

\begin{proof}
At first, we let $n=q+1$ and
$$H:=\left(\begin{array}{cccc}
\alpha^{4} & (\alpha^{4})^{2} & \cdots & (\alpha^{4})^{n}=1 \\
\alpha^{5} & (\alpha^{5})^{2} & \cdots & (\alpha^{5})^{n}=1 \\
\end{array}\right)$$
with $\alpha$ being a generator of $U_{q+1}$. Then $H$ is the
parity-check matrix of $\mathcal{C}_{(q, q+1, 3, 4)}$, i.e.,
$$\mathcal{C}_{(q, q+1, 3, 4)}
=\{{\bf c}\in\mathbb{F}_{q}^{q+1}|H{\bf c}^T={\bf 0}\}.$$
Writing $h_{i}=\alpha^{i}$ for $i\in \langle n\rangle$
gives that
$$H=\left(\begin{array}{cccc}
h_{1}^{4} & h_{2}^{4} & \cdots & h_{n}^{4} \\
h_{1}^{5} & h_{2}^{5} & \cdots & h_{n}^{5} \\
\end{array}\right).$$

We first treat the necessity part. Let
$\dim\mathcal{C}(i_{1}, ..., i_{n-k+1=5})=1$ hold for
all subsets $\{i_{1}, ..., i_{5}\}\subseteq\langle
n\rangle$. For arbitrary two distinct elements
$x, y\in U_{q+1}\backslash\{1:=1_{\mathbb{F}_{q^{2}}}\}
\subseteq \mathbb{F}_{q^{2}}$, let $j_{1}, j_{2}\in
\langle n\rangle \backslash\{n\}$ be such that $x=h_{j_1}$
and $y=h_{j_2}$. Assume that there are two distinct indexes
$i_{1}, i_{2}\in\langle n\rangle\backslash
\{j_{1}, j_{2}, n\}$ such that the system of
homogeneous equations
$$\left(\begin{array}{cccc}
h_{j_{1}}^{4} & h_{j_{2}}^{4} & h_{i_{1}}^{4}& 1 \\
h_{j_{1}}^{5} & h_{j_{2}}^{5} & h_{i_{1}}^{5}& 1 \\
\end{array}
\right){\bf X}={\bf 0}$$
has a  nonzero solution ${\bf X} =(x_1, x_2, x_3, x_4)^T$
in ${\mathbb{F}}^{4}_{q}$ and the system of homogeneous equations
$$\left(\begin{array}{cccc}
h_{j_{1}}^{4} & h_{j_{2}}^{4} & h_{i_{2}}^{4}& 1 \\
h_{j_{1}}^{5} & h_{j_{2}}^{5} & h_{i_{2}}^{5}& 1 \\
\end{array}\right){\bf Y}={\bf 0}$$
has a nonzero solution ${\bf Y}=(y_1, y_2, y_3, y_4)^T$
in ${\mathbb{F}}^{4}_{q}$. Then we can pick two nonzero
vectors ${\bf c}, {\bf c}'\in{\mathbb{F}}^{n}_{q}$
such that
$${\bf c}_{j_{1}}=x_1,{\bf c}_{j_{2}}=x_2,
{\bf c}_{i_{1}}=x_3,{\bf c}_{n}=x_4 \ {\rm and}
\ {\bf c}_j=0 \ \forall j\in\langle n\rangle
\setminus\{j_{1},j_{2},i_{1},n\},$$
and
$${\bf c}'_{j_{1}}=y_1,{\bf c}'_{j_{2}}=y_2,
{\bf c}'_{i_{2}}=y_3,{\bf c}'_{n}=y_4
\ {\rm and} \ {\bf c}_j=0 \ \forall j\in
\langle n\rangle\setminus\{j_1, j_2, i_2, n\},$$
respectively. So we have
${\rm supp}({\bf c})\subseteq\{j_{1},j_{2},i_{1},n\}
\subseteq\{j_{1}, j_{2}, i_{1}, i_{2}, n\}
\subseteq\langle n\rangle$ and
${\rm supp}({\bf c}')\subseteq\{j_{1}, j_{2},i_{2}, n\}
\subseteq \{j_{1}, j_{2}, i_{1}, i_{2}, n\}
\subseteq\langle n\rangle$.
Letting $H_{i}=\left(\begin{array}{c}
h_{i}^{4}\\
h_{i}^{5}
\end{array}\right)$ as the $i$-th column
vector of $H$ for $i\in \langle n\rangle$ then gives us that
\begin{align*}
H{\bf c}^T=&\sum_{i=1}^{n}{\bf c}_{i}H_{i}\\
=&\sum_{i\in\langle n\rangle\backslash\{j_{1},j_{2},
i_{1},n\}}{\bf c}_{i}H_{i}+{\bf c}_{j_{1}}H_{j_{1}}
+{\bf c}_{j_{2}}H_{j_{2}}+{\bf c}_{i_{1}}H_{i_{1}}+{\bf c}_{n}H_{n}\\
=&{\bf 0}+{\bf c}_{j_{1}}H_{j_{1}}+{\bf c}_{j_{2}}H_{j_{2}}
+{\bf c}_{i_{1}}H_{i_{1}}+{\bf c}_{n}H_{n}\\
=&x_{1}H_{j_{1}}+x_{2}H_{j_{2}}+x_{3}H_{i_{1}}+x_{4}H_{n}\\
=&\left(
  \begin{array}{cccc}
    h_{j_{1}}^{4} & h_{j_{2}}^{4} & h_{i_{1}}^{4} & 1 \\
    h_{j_{1}}^{5} & h_{j_{2}}^{5} & h_{i_{1}}^{5} & 1 \\
  \end{array}
  \right)(x_1, x_2, x_3, x_4)^T\\
=& {\bf 0}
\end{align*}
and
\begin{align*}
H({\bf c}')^T=&\sum_{i=1}^{n}{\bf c}'_{i}H_{i}\\
=&\sum_{i\in\langle n\rangle
\backslash\{j_{1},j_{2},i_{2},n\}}{\bf c}'_{i}H_{i}
+{\bf c}'_{j_{1}}H_{j_{1}}+{\bf c}'_{j_{2}}H_{j_{2}}
+{\bf c}'_{i_{2}}H_{i_{2}}+{\bf c}'_{n}H_{n}\\
=&{\bf 0}+{\bf c}'_{j_{1}}H_{j_{1}}+{\bf c}'_{j_{2}}
H_{j_{2}}+{\bf c}'_{i_{2}}H_{i_{2}}+{\bf c}'_{n}H_{n}\\
=&y_{1}H_{j_{1}}+y_{2}H_{j_{2}}+y_{3}H_{i_{2}}+y_{4}H_{n}\\
=&\left(\begin{array}{cccc}
h_{j_{1}}^{4} & h_{j_{2}}^{4} & h_{i_{2}}^{4}& 1 \\
h_{j_{1}}^{5} & h_{j_{2}}^{5} & h_{i_{2}}^{5}& 1 \\
\end{array}
\right)(y_1, y_2, y_3, y_4)^T\\
=&{\bf 0}.
\end{align*}
This implies that ${\bf c}, {\bf c}'\in {\mathcal{C}}
=\mathcal{C}_{(q, q+1, 3, 4)}$. Thus $w({\bf c})
=|{\rm supp}({\bf c})|\ge d({\mathcal C})=4$ and
$w({\bf c}')=|{\rm supp}({\bf c}')|\ge d({\mathcal C})=4$.
But $|{\rm supp}({\bf c})|\le 4$ and $|{\rm supp}
({\bf c}')|\le 4$. Hence we must have ${\rm supp}
({\bf c})=\{j_{1}, j_{2}, i_{1}, n\}
\subseteq \{j_{1}, j_{2}, i_{1}, i_{2}, n\}$ and
${\rm supp}({\bf c}')=\{j_{1}, j_{2}, i_{2}, n\}
\subseteq \{j_{1}, j_{2}, i_{1}, i_{2}, n\}$.
One then deduces that ${\bf c}, {\bf c}'\in
\mathcal{C}(j_{1}, j_{2}, i_{1}, i_{2}, n)$.
It then follows from  ${\rm supp}({\bf c})
\neq {\rm supp}({\bf c}')$ that ${\bf c}$ and
${\bf c}'$ are linear independent, from which
one derives that
$\dim\mathcal{C}(j_{1}, j_{2}, i_{1}, i_{2}, n)\ge 2$.
This contradicts with the hypothesis
$\dim\mathcal{C}(j_{1}, j_{2}, i_{1}, i_{2}, n)=1$.
Hence for arbitrary two distinct elements
$x, y\in U_{q+1}\backslash\{1:=1_{\mathbb{F}_{q^{2}}}\}
\subseteq \mathbb{F}_{q^{2}}$, there exists at most one element
$z\in U_{q+1}\backslash\{x, y, 1\}\subseteq{\mathbb{F}_{q^{2}}}$
such that the system (\ref{3.1}) of homogeneous equations
has a nonzero solution ${\bf X}$ in ${\mathbb{F}}^{4}_{q}$.
The necessity part is proved.

Consequently, we show the sufficiency part.
First of all, Lemma \ref{lem7} tells us that
$\dim{\mathcal C}(i_1, ..., i_5)\in\{1, 2\}$
for all subsets $\{i_1, \cdots, i_5\}
\subseteq\langle n\rangle$. We assume that
there exists a subset
$\{j_{1}, j_{2}, j_{3}, i_{1}, i_{2}\}\subseteq
\langle n\rangle$ such that
$\dim\mathcal{C}(j_{1}, j_{2}, j_{3}, i_{1}, i_{2})=2$.
Then there exist two linearly independent nonzero
codewords ${\bf c}_{1}, {\bf c}_{2}\in
\mathcal{C}(j_{1}, j_{2}, j_{3}, i_{1}, i_{2})$.
Thus $w({\bf c}_1), w({\bf c}_2)\in\{4, 5\}$.
We claim that there are two linearly independent
nonzero codewords ${\bf c}_{3}$ and ${\bf c}_{4}$
in $\mathcal{C}(j_{1}, j_{2}, j_{3}, i_{1}, i_{2})$
with weight 4. If $w({\bf c}_1)=4$ and $w({\bf c}_2)=4$,
then letting ${\bf c}_3:={\bf c}_1\in\mathcal{C}
(j_{1}, j_{2}, j_{3}, i_{1}, i_{2})$
and ${\bf c}_4:={\bf c}_2\in\mathcal{C}
(j_{1}, j_{2}, j_{3}, i_{1}, i_{2})$ gives
us the two claimed linearly independent nonzero
codewords ${\bf c}_3$ and ${\bf c}_4$ in
$\mathcal{C}(j_{1}, j_{2}, j_{3}, i_{1}, i_{2})$
with $w({\bf c}_3)=w({\bf c}_4)=4$. When $w({\bf c}_1)=5$
or $w({\bf c}_2)=5$, in order to finish the proof of
the claim, one needs only to consider the
following three cases.

{\it Case 1.} $w({\bf c}_1)=w({\bf c}_2)=5$. Then
${\rm supp}({\bf c}_1)={\rm supp}({\bf c}_2)
=\{j_{1}, j_{2}, j_{3}, i_{1}, i_{2}\}$.
Obviously, $({\bf c}_1)_{i_{1}}\neq0$ and
$({\bf c}_2)_{i_{1}}\neq0$. Let
${\bf c}_3={\bf c}_1-(({\bf c}_1)_{i_{1}}
({\bf c}_2)^{-1}_{i_{1}}){\bf c}_2
\in\mathcal{C}(j_{1}, j_{2}, j_{3}, i_{1}, i_{2})$.
Thus
\begin{align*}
({\bf c}_3)_{i_{1}}=&({\bf c}_1-(({\bf c}_1)_{i_{1}}
({\bf c}_2)^{-1}_{i_{1}}){\bf c}_2)_{i_{1}}\\
=&({\bf c}_1)_{i_{1}}-((({\bf c}_1)_{i_{1}}
({\bf c}_2)^{-1}_{i_{1}}){\bf c}_2)_{i_{1}}\\
=& ({\bf c}_1)_{i_{1}}-({\bf c}_1)_{i_{1}}
({\bf c}_2)^{-1}_{i_{1}}({\bf c}_2)_{i_{1}}\\
=& 0.
\end{align*}
It implies that ${\rm supp}({\bf c}_3)
\subseteq\{j_{1}, j_{2}, j_{3}, i_{2}\}$ and so
we have $w({\bf c}_3)={\rm supp}({\bf c}_3)\le 4$.
Since ${\bf c}_1$ and ${\bf c}_2$ are linearly
independent nonzero codewords, ${\bf c}_3$ must
be a nonzero codeword. This implies that
$4=d(\mathcal{C})\leq w({\bf c}_3)\le 4$.
Hence we find a codeword ${\bf c}_3\in \mathcal{C}
(j_{1}, j_{2}, j_{3}, i_{1}, i_{2})$
with $w({\bf c}_3)=4$. Likewise, since
$({\bf c}_2)_{i_{2}}\neq0$, letting
${\bf c}_4:={\bf c}_2-(({\bf c}_2)_{i_{2}}
({\bf c}_1)^{-1}_{i_{2}}){\bf c}_1$ gives us
a codeword ${\bf c}_4\in \mathcal{C}
(j_{1}, j_{2}, j_{3}, i_{1}, i_{2})$
with ${\rm supp}({\bf c}_4)=\{j_{1}, j_{2},
j_{3}, i_{1}\}$ and $w({\bf c}_4)=4$.
Moreover,
${\rm supp}({\bf c}_3)=\{j_{1}, j_{2}, j_{3},
i_{2}\}\neq\{j_{1}, j_{2}, j_{3}, i_{1}\}
={\rm supp}({\bf c}_4)$
which implies that ${\bf c}_3$ and ${\bf c}_4$
are linearly independent nonzero codewords.
Therefore we obtain two linearly independent
nonzero codewords ${\bf c}_3, {\bf c}_4
\in\mathcal{C}(j_{1}, j_{2}, j_{3}, i_1,
i_{2})$ with $w({\bf c}_3)=w({\bf c}_4)=4$
as desired. The claim is proved in this case.

{\it Case 2.} $w({\bf c}_1)=5$ and $w({\bf c}_2)=4$.
Then ${\rm supp}({\bf c}_1)=\{j_{1}, j_{2}, j_{3}, i_{1}, i_{2}\}$
and ${\rm supp}({\bf c}_2)\subseteq\{j_{1}, j_{2}, j_{3}, i_{1}, i_{2}\}$.
Evidently, $({\bf c}_1)_{i_{1}}\neq0$ and we may
let $({\bf c}_2)_{i_{2}}\neq0$. Picking
${\bf c}_3={\bf c}_1-(({\bf c}_1)_{i_{2}}({\bf c}_2)
^{-1}_{i_{2}}){\bf c}_2$ gives that
${\rm supp}({\bf c}_3)\subseteq\{j_{1}, j_{2}, j_{3}, i_{1}, i_{2}\}$
and
\begin{align*}
({\bf c}_3)_{i_{2}}=&({\bf c}_1-(({\bf c}_1)_{i_{2}}
({\bf c}_2)^{-1}_{i_{2}}){\bf c}_2)_{i_{2}}\\
=& ({\bf c}_1)_{i_{2}}-(({\bf c}_1)_{i_{2}}
({\bf c}_2)^{-1}_{i_{2}}){\bf c}_2)_{i_{2}}\\
=& ({\bf c}_1)_{i_{2}}-({\bf c}_1)_{i_{2}}
({\bf c}_2)^{-1}_{i_{2}}({\bf c}_2)_{i_{2}}\\
=& 0.
\end{align*}
Hence ${\rm supp}({\bf c}_3)\subseteq\{j_{1}, j_{2}, j_{3}, i_{1}\}$.
As ${\bf c}_1$ and ${\bf c}_2$ are linearly independent
nonzero codewords, ${\bf c}_3$ must be a nonzero codeword.
This implies that $4=d(\mathcal{C})\leq w({\bf c}_3)
=|{\rm supp}({\bf c}_3)|\le 4$. So one gets a codeword
${\bf c}_3\in \mathcal{C}(j_{1}, j_{2}, j_{3}, i_{1}, i_{2})$
with $w({\bf c}_3)=4$. Hence we find two linearly
independent nonzero codewords ${\bf c}_3$ and
${\bf c}_4={\bf c}_2$ in $\mathcal{C}(j_{1}, j_{2},
j_{3}, i_{1}, i_{2})$ with $w({\bf c}_3)=w({\bf c}_4)=4$.
The claim is true in this case.

{\it Case 3.} $w({\bf c}_1)=4$ and $w({\bf c}_2)=5$.
Then by the symmetry, letting ${\bf c}_4={\bf c}_2
-(({\bf c}_2)_{i_{2}}({\bf c}_2)^{-1}_{i_{2}}){\bf c}_1$
gives us two linearly independent nonzero codewords
${\bf c}_3:={\bf c}_1$ and ${\bf c}_4$ in
$\mathcal{C}(j_{1}, j_{2}, j_{3}, i_{1}, i_{2})$
with $w({\bf c}_3)=w({\bf c}_4)=4$.
Hence the claim is proved in this case.

By the above claim, we always have two linearly
independent nonzero codewords ${\bf c}_3$ and
${\bf c}_4$ in $\mathcal{C}(j_{1}, j_{2}, j_{3},
i_{1}, i_{2})$ with $w({\bf c}_3)=w({\bf c}_4)=4$.
Suppose that ${\rm supp}({\bf c}_3)={\rm supp}({\bf c}_4)$.
One may write ${\rm supp}({\bf c}_3)
={\rm supp}({\bf c}_4)=\{j_{1}, j_{2}, j_{3}, i_{1}\}$.
Let ${\bf c}_{5}:={\bf c}_{3}-(({\bf c}_3)_{i_1}
({\bf c}_{4})_{i_{1}}^{-1}){\bf c}_{4}$.
Then $({\bf c}_5)_{i_1}=0$ and so ${\rm supp}
({\bf c}_5)\subseteq\{j_{1}, j_{2}, j_{3}\}$.
Since ${\bf c}_3$ and ${\bf c}_4$ are linearly
independent in $\mathcal{C}(j_{1}, j_{2}, j_{3},
i_{1}, i_{2})$, it follows that ${\bf c}_{5}$
is a nonzero codeword. So $w({\bf c}_{5})\le 3$
which contradicts with the fact
$w({\bf c}_{5})\ge d(\mathcal{C})=4$.
Thus we must have ${\rm supp}({\bf c}_{3})
\neq {\rm supp}({\bf c}_{4})$. However,
$w({\bf c}_3)=w({\bf c}_4)=4$. Hence
$|({\rm supp}({\bf c}_{3})\cap{\rm supp}
({\bf c}_{4})|\le 3$.

On the other hand, since
${\bf c}_3, {\bf c}_4\in\mathcal{C}
(j_{1}, j_{2}, j_{3}, i_{1}, i_{2})$, we
have ${\rm supp}({\bf c}_3)\cup{\rm supp}
({\bf c}_4)\subseteq\{j_{1}, j_{2}, j_{3},
i_1, i_2\}$. Thus $|{\rm supp}({\bf c}_3)
\cup{\rm supp}({\bf c}_4)|\le 5$. It then follows
from the inclusion-exclusion principle that
\begin{align*}
|{\rm supp}({\bf c}_{3})\cap {\rm supp}
({\bf c}_{4})|=&  |{\rm supp}({\bf c}_{3})|
+|{\rm supp}({\bf c}_{4})|-|{\rm supp}
({\bf c}_{3})\cup {\rm supp}({\bf c}_{4})|\\
 \geq&  4+4-5=3.
\end{align*}
This concludes that $|{\rm supp}({\bf c}_{3})
\cap {\rm supp}({\bf c}_{4})|=3$. One may write
${\rm supp}({\bf c}_{3})=\{j_{1}, j_{2}, j_{3},
i_{1}\}$ and ${\rm supp}({\bf c}_{4})=\{j_{1},
j_{2}, j_{3}, i_{2}\}$. Then
${\bf X}:=(({\bf c}_{3})_{j_{1}}, ({\bf c}_{3})_{j_{2}},
({\bf c}_{3})_{j_{3}}, ({\bf c}_{3})_{i_{1}})^T$
is a solution of the following system of homogeneous equations
$$\left(\begin{array}{cccc}
h_{j_{1}}^{4} & h_{j_{2}}^{4} & h_{j_{3}}^{4}& h_{i_{1}}^{4} \\
h_{j_{1}}^{5} & h_{j_{2}}^{5} & h_{j_{3}}^{5}& h_{i_{1}}^{5} \\
\end{array}\right){\bf X}={\bf 0}$$
and ${\bf Y}:=(({\bf c}_{4})_{j_{1}},
({\bf c}_{4})_{j_{2}}, ({\bf c}_{4})_{j_{3}},
({\bf c}_{4})_{i_{2}})^T$ is a solution
of the system of homogeneous equations
$$\left(\begin{array}{cccc}
h_{j_{1}}^{4} & h_{j_{2}}^{4} & h_{j_{3}}^{4}& h_{i_{2}}^{4} \\
h_{j_{1}}^{5} & h_{j_{2}}^{5} & h_{j_{3}}^{5}& h_{i_{2}}^{5} \\
\end{array}\right){\bf Y}={\bf 0}.$$
It then follows that
$$\left(\begin{array}{cccc}
(h_{j_{1}}h_{j_{3}}^{-1})^{4} & (h_{j_{2}}h_{j_{3}}^{-1})^{4} & 1 & (h_{i_{1}}h_{j_{3}}^{-1})^{4} \\
(h_{j_{1}}h_{j_{3}}^{-1})^{5} & (h_{j_{2}}h_{j_{3}}^{-1})^{5} & 1 & (h_{i_{1}}h_{j_{3}}^{-1})^{5} \\
\end{array}\right){\bf X}
=\left(\begin{array}{c}
  0\times h_{j_{3}}^{-4}\\
  0\times h_{j_{3}}^{-5}
\end{array}\right)={\bf 0}$$
and
$$\left(
  \begin{array}{cccc}
    (h_{j_{1}}h_{j_{3}}^{-1})^{4} &
    (h_{j_{2}}h_{j_{3}}^{-1})^{4} & 1
    & (h_{i_{2}}h_{j_{3}}^{-1})^{4} \\
    (h_{j_{1}}h_{j_{3}}^{-1})^{5}
    & (h_{j_{2}}h_{j_{3}}^{-1})^{5} & 1
    & (h_{i_{2}}h_{j_{3}}^{-1})^{5} \\
  \end{array}
  \right){\bf Y}=
  \left(\begin{array}{c}
  0\times h_{j_{3}}^{-4}\\
  0\times h_{j_{3}}^{-5}
  \end{array}\right)={\bf 0}.$$

Let $x_0=h_{j_{1}}h_{j_{3}}^{-1}, y_0=h_{j_{2}}h_{j_{3}}^{-1},
z_0=h_{i_{1}}h_{j_{3}}^{-1}$ and $z'_0=h_{i_{2}}h_{j_{3}}^{-1}$.
Then $x_0, y_0, z_0, z_0'\in U_{q+1}$ are pairwise distinct,
and both of the systems of homogeneous equations
$$\left(\begin{array}{cccc}
x_0^{4} & y_0^{4} & z_0^{4} & 1 \\
x_0^{5} & y_0^{5} & z_0^{5} & 1 \\
\end{array}
\right){\bf X}={\bf 0}$$
and
$$\left(\begin{array}{cccc}
x_0^{4} & y_0^{4} & z_0'^{4} & 1 \\
x_0^{5} & y_0^{5} & z_0'^{5} & 1 \\
\end{array}\right){\bf Y}={\bf 0}$$
have nonzero solutions in ${\mathbb{F}}^{4}_{q}$,
which contradicts with the hypothesis.
Thus the assumption is not true. Hence we must have
$\dim\mathcal{C}(j_{1}, j_{2}, j_{3}, i_{1}, i_{2})=1$
for any subset $\{j_{1}, j_{2}, j_{3}, i_{1}, i_{2}\}
\subseteq\langle n\rangle$. The sufficiency part
is proved.

This concludes the proof of Lemma \ref{lem3.1}.
\end{proof}

By the theory of the system of homogeneous equations
over a field, one knows that for arbitrary elements
$x, y, z\in U_{q+1}\setminus\{1\}$, the system
(\ref{3.1}) of homogeneous equations always has a
nonzero solution $X$ in $\mathbb{F}_{q^2}^4$ when the
characteristics of the field is equal to 3. But it
does not guarantee that the system (\ref{3.1}) of
homogeneous equations holds a nonzero solution $X$
in $\mathbb{F}_{q^2}^4$.

The following result is known and is given in
\cite{[GYZZ-FFA22]}.

\begin{lem}\label{lem3} \cite{[GYZZ-FFA22]}
Let $q=3^{m}$ with $m$ being an odd integer.
Let $U_{q+1}$ denote the group consisting of
all $(q+1)$-th roots of unity in
${\mathbb{F}}_{q^{2}}$. If $x,y$ and $z$ are three
pairwise distinct elements in $U_{q+1}$, then
$$\det\left(\begin{array}{ccc}
    x^{4} & y^{4} & z^{4} \\
    x^{5} & y^{5} & z^{5} \\
    x^{-5} & y^{-5} & z^{-5} \\
\end{array}\right)\neq0.$$
\end{lem}

\begin{lem}\label{lem10}
Let $q=3^{m}$ with $m$ being an odd integer.
Let $x,y,z\in U_{q+1}\backslash\{1:=1_{\mathbb{F}_{q^{2}}}\}
\subseteq \mathbb{F}_{q^{2}}$ be
arbitrary three pairwise distinct elements.
Then each of the following is true:

{\rm (i).} The system (\ref{3.1}) of homogeneous equations
has a nonzero solution in ${\mathbb{F}}^{4}_{q}$ if and only
if the following system of homogeneous equations
\begin{align}\label{3.2}
\left(\begin{array}{cccc}
 x^{4} & y^{4} & z^{4} & 1 \\
 x^{5} & y^{5} & z^{5} & 1 \\
 x^{-5} & y^{-5} & z^{-5} & 1 \\
 x^{-4} & y^{-4} & z^{-4} & 1 \\
  \end{array}
\right){\bf X}={\bf 0}
\end{align}
has a nonzero solution in $\mathbb{F}_{q}^{4}$.

{\rm (ii).} The system (\ref{3.2}) has a nonzero
solution in $\mathbb{F}_{q}^{4}$ if and only if
the following holds in $\mathbb{F}_{q^{2}}$:
\begin{align}\label{3.3}
\det\left(\begin{array}{cccc}
 x^{4} & y^{4} & z^{4} & 1 \\
 x^{5} & y^{5} & z^{5} & 1 \\
 x^{-5} & y^{-5} & z^{-5} & 1 \\
 x^{-4} & y^{-4} & z^{-4} & 1 \\
\end{array}\right)=0.
\end{align}
\end{lem}

\begin{proof}
(i). We begin with the proof of the necessity.
Let the system (\ref{3.1}) of the homogeneous equations
hold a nonzero solution $(a,b,c,d)\in\mathbb{F}_{q}^{4}$.
Then for $i\in\{4,5\}$, we have
$$ax^{i}+by^{i}+cz^{i}+d=0.$$
Since $x,y,z\in U_{q+1}\setminus\{1\}$ and $a,b,c,d\in
{\mathbb{F}_{q}}$, we have $x^{q+1}=y^{q+1}=z^{q+1}=1$
and $a^q=a, b^q=b, c^q=c$ and $d^q=d$.
It follows that $x^{-1}=x^q, y^{-1}=y^q, z^{-1}=z^q$.
Hence for $i\in\{4, 5\}$, we have
$$ax^{-i}+by^{-i}+cz^{-i}+d
=(ax^{i}+by^{i}+cz^{i}+d)^{q}=0^q=0.$$
It then follows that $(a,b,c,d)\in\mathbb{F}_{q}^{4}$
is also a nonzero solution in $\mathbb{F}_{q}^{4}$
of the system (\ref{3.2}) of homogeneous equations.
So the necessity part is proved.

Consequently, we show the sufficiency part.
Let the system (\ref{3.2}) of homogeneous equations
have a nonzero solution $(a,b,c,d)\in\mathbb{F}_{q}^{4}
\subseteq\mathbb{F}_{q^2}^{4}$. Then for $i\in\{4,5\}$,
we have $ax^{i}+by^{i}+cz^{i}+d=0.$
In other words, $(a,b,c,d)\in\mathbb{F}_{q}^{4}$
is a nonzero solution of the system (\ref{3.1}) of
homogeneous equations. Therefore the sufficiency
part is proved. So part (i) is proved.

(ii). First, we show the necessity part.
Let the system (\ref{3.2}) of homogeneous equations
have a nonzero solution in $\mathbb{F}_q^4\subseteq
\mathbb{F}_{q^2}^4$. Then by the theory of linear
algebra over finite fields, one knows that (\ref{3.3})
is true over the finite field $\mathbb{F}_{q^2}$.
Hence the necessity is proved.

Conversely, we consider the sufficiency part.
Let (\ref{3.3}) hold over the finite field
$\mathbb{F}_{q^2}$. Then the system (\ref{3.2})
of homogeneous equations has a nonzero solution
$(a,b,c,d)\in\mathbb{F}_{q^2}^{4}$. So for
$i\in\{4, 5\}$, we have
$$ax^{i}+by^{i}+cz^{i}+d=0$$
and
$$ax^{-i}+by^{-i}+cz^{-i}+d=0.$$
Since $(a,b,c,d)\in\mathbb{F}_{q^2}^{4}$ is nonzero,
at least one of $a, b, c$ and $d$ is a nonzero
element in $\mathbb{F}_{q^2}$. WLOG, one may let
$a\ne 0$. Then for $i\in\{4, 5\}$, one has
$$x^{i}+\frac{b}{a}y^{i}+\frac{c}{a}z^{i}+\frac{d}{a}=0$$
and
$$x^{-i}+\frac{b}{a}y^{-i}+\frac{c}{a}z^{-i}+\frac{d}{a}=0.$$
Write $\bar b=\frac{b}{a}, \bar c=\frac{c}{a}$
and $\bar d=\frac{d}{a}$. Then
$(1, \bar b, \bar c, \bar d)\in\mathbb{F}_{q^2}^{4}$
is solution of (\ref{3.2}).

From the condition that $x,y,z\in U_{q+1}\backslash\{1\}$,
one derives that $x=x^{-q}, y=y^{-q}$ and $z=z^{-q}$.
It then follows that for $i\in\{4, 5\}$, we have
$$x^{i}+\bar b^qy^{i}+\bar c^qz^{i}+\bar d^q
=x^{-qi}+\bar b^qy^{-qi}+\bar c^qz^{-qi}+\bar d^q
=(x^{-i}+\bar by^{-i}+\bar cz^{-i}+\bar d)^q
=0^q=0$$
and
$$x^{-i}+\bar b^qy^{-i}+\bar c^qz^{-i}+\bar d^q
=x^{qi}+\bar b^qy^{qi}+\bar c^qz^{qi}+\bar d^q
=(x^{i}+\bar by^{i}+\bar cz^{i}+\bar d)^q
=0^q=0.$$
This implies that $(1, \bar b^q, \bar c^q, \bar d^q)
\in\mathbb{F}_{q^2}^{4}$ is a nonzero solution
of (\ref{3.2}).

On the other hand, since (3.3) is true, we have
\begin{align}\label{3.4}
{\rm rank}\left(\begin{array}{cccc}
   x^{4} & y^{4} & z^{4} & 1 \\
 x^{5} & y^{5} & z^{5} & 1 \\
 x^{-5} & y^{-5} & z^{-5} & 1 \\
 x^{-4} & y^{-4} & z^{-4} & 1 \\
  \end{array}
\right)\le 3.
\end{align}
But Lemma \ref{lem3} tells us that
$${\rm rank}\left(
  \begin{array}{ccc}
 x^{4} & y^{4} & z^{4} \\
 x^{5} & y^{5} & z^{5} \\
 x^{-5} & y^{-5} & z^{-5} \\
  \end{array}
\right)= 3.$$
Then
\begin{align}\label{3.5}
{\rm rank}\left(
  \begin{array}{cccc}
 x^{4} & y^{4} & z^{4} & 1 \\
 x^{5} & y^{5} & z^{5} & 1 \\
 x^{-5} & y^{-5} & z^{-5} & 1 \\
 x^{-4} & y^{-4} & z^{-4} & 1 \\
  \end{array}
\right)\geq {\rm rank}\left(
  \begin{array}{ccc}
 x^{4} & y^{4} & z^{4} \\
 x^{5} & y^{5} & z^{5} \\
 x^{-5} & y^{-5} & z^{-5} \\
  \end{array}
\right)= 3.
\end{align}
Combining (\ref{3.4}) and (\ref{3.5}), one can
conclude that
$${\rm rank}\left(\begin{array}{cccc}
x^{4} & y^{4} & z^{4} & 1 \\
x^{5} & y^{5} & z^{5} & 1 \\
x^{-5} & y^{-5} & z^{-5} & 1 \\
x^{-4} & y^{-4} & z^{-4} & 1 \\
\end{array}\right)=3.$$
It then follows that the dimension of the
solution space over $\mathbb{F}_{q^{2}}$ of
the system (\ref{3.2}) of homogeneous equations
is equal to $4-{\rm rank}(A)=4-3=1$.
This infers that the two solutions
$(1, \bar b, \bar c, \bar d)^T$ and $(1, \bar b^{q},
\bar c^{q}, \bar d^{q})^T$ are linearly dependent
which implies that there is a nonzero element
$l\in\mathbb{F}_{q^2}^*$ such that
$(1, \bar b, \bar c, \bar d)^T=l(1, \bar b^{q},
\bar c^{q}, \bar d^{q})^T$. It follows that $l=1$
and so we have $\bar b=\bar b^{q},
\bar c=\bar c^{q}$ and $\bar d=\bar d^{q}$.
This implies that each of $\bar b, \bar c$ and
$\bar d$ belongs to $\mathbb{F}_q$. Hence
$(1, \bar b, \bar c, \bar d)$ is a nonzero
solution in $\mathbb{F}_q^4$ of (\ref{3.2}).
The sufficiency part is proved.

This completes the proof of Lemma \ref{lem10}.
\end{proof}

\begin{lem}\label{lem11}
Let $q=3^{m}$ with $m$ being a positive integer.
Let $x,y,z\in U_{q+1}\backslash\{1\}$ be
three pairwise distinct elements. Then
$$\det\left(\begin{array}{cccc}
x^{4} & y^{4} & z^{4} & 1 \\
x^{5} & y^{5} & z^{5} & 1 \\
x^{-5} & y^{-5} & z^{-5} & 1 \\
x^{-4} & y^{-4} & z^{-4} & 1 \\
\end{array}\right)
=\frac{XYZ}{((X+1)(Y+1)(Z+1))^{5}}
\det\left(\begin{array}{cc}
Y^9-X^{9} & Z^9-X^{9}\\
Y^8-X^{8}& Z^{8}-X^{8}
\end{array}\right),$$
where $X:=x-1$, $Y:=y-1$ and $Z:=z-1$.
Moreover, it holds in $\mathbb{F}_{q^2}$ that
$$\det\left(\begin{array}{cccc}
x^{4} & y^{4} & z^{4} & 1 \\
 x^{5} & y^{5} & z^{5} & 1 \\
 x^{-5} & y^{-5} & z^{-5} & 1 \\
 x^{-4} & y^{-4} & z^{-4} & 1 \\
 \end{array}
\right)=0$$
if and only if it holds in $\mathbb{F}_{q^2}$ that
\begin{equation*}
\det\left(\begin{array}{cc}
Y^{9}-X^{9} & Z^{9}-X^{9} \\
Y^{8}-X^{8} & Z^{8}-X^{8}
\end{array}\right)=0.
\end{equation*}
\end{lem}

\begin{proof}
Let
$$A=\left(
  \begin{array}{cccc}
 x^{4} & y^{4} & z^{4} & 1 \\
 x^{5} & y^{5} & z^{5} & 1 \\
 x^{-5} & y^{-5} & z^{-5} & 1 \\
 x^{-4} & y^{-4} & z^{-4} & 1 \\
  \end{array}
\right).$$
Then
\begin{align*}
\det(A)=&\det\left(\begin{array}{ccc}
        x^{4}-x^{-4} & y^{4}-y^{-4} & z^{4}-z^{-4}\\
        x^{5}-x^{-4} & y^{5}-y^{-4} & z^{5}-z^{-4}\\
        x^{-5}-x^{-4} & y^{-5}-y^{-4} & z^{-5}-z^{-4}
      \end{array}
\right)\\
=&(xyz)^{-4}\det\left(\begin{array}{ccc}
x^{8}-1 & x^{8}-1 & z^{8}-1\\
x^{9}-1 & y^{9}-1 & z^{9}-1\\
x^{-1}-1 & y^{-1}-1 & z^{-1}-1
\end{array}\right).
\end{align*}
Since $q=3^{m}$ with $m\ge 1$ being an integer, one has
$x^9-1=(x-1)^9, y^9-1=(y-1)^9$ and $z^9-1=(z-1)^9$.
Moreover, we have
$$\frac{x^{8}-1}{x-1}=(1+x+\cdots+x^7+x^8)-x^8
=\frac{(x-1)^{9}}{x-1}-x^8=(x-1)^{8}-x^8.$$
Likewise, we have
$$\frac{y^{8}-1}{y-1}=(y-1)^{8}-y^8,
\frac{z^{8}-1}{z-1}=(z-1)^{8}-z^8.$$
It follows that
\begin{align*}
\det(A)
=&(xyz)^{-4}(x-1)(y-1)(z-1)\\
&\times\det\left(\begin{array}{ccc}
        (x-1)^{8}-x^8 & (y-1)^{8}-y^8 & (z-1)^{8}-z^8\\
        (x-1)^{8} & (y-1)^{8} & (z-1)^{8}\\
        -x^{-1} & -y^{-1} & -z^{-1}
      \end{array}
\right)\\
=&(xyz)^{-4}(x-1)(y-1)(z-1)\det\left(\begin{array}{ccc}
        x^8 & y^8 & z^8\\
        (x-1)^{8} & (y-1)^{8} & (z-1)^{8}\\
       x^{-1} & y^{-1} & z^{-1}
      \end{array}
\right)\\
=& (xyz)^{-4}(x-1)(y-1)(z-1)\\
&\times\det\left(\begin{array}{ccc}
        x^8 & y^8-\frac{x}{y}x^{8}
        & z^8-\frac{x}{z}x^{8}\\
        (x-1)^{8} & (y-1)^8-\frac{x}{y}(x-1)^{8}
        & (z-1)^8-\frac{x}{z}(x-1)^{8}\\
       x^{-1} & 0 & 0
      \end{array}
\right)\\
=& x^{-5}(yz)^{-4}(x-1)(y-1)(z-1)\\
&\times\det\left(\begin{array}{cc}
y^8-x^{9}y^{-1} & z^8-x^{9}z^{-1}\\
(y-1)^8-x(x-1)^{8}y^{-1}
&(z-1)^8-x(x-1)^{8}z^{-1}
\end{array}\right)\\
=&(xyz)^{-5}(x-1)(y-1)(z-1)\\
&\times\det\left(\begin{array}{cc}
y^9-x^{9} & z^9-x^{9}\\
y(y-1)^8-x(x-1)^{8}& z(z-1)^8-x(x-1)^{8}
\end{array}\right).
\end{align*}
But
$$y(y-1)^{8}=(y-1)^{9}+(y-1)^{8}=y^{9}-1+(y-1)^{8}$$
and
$$z(z-1)^{8}=z^{9}-1+(z-1)^{8}.$$
Then one arrives at
\begin{align*}
&\det(A)=(xyz)^{-5}(x-1)(y-1)(z-1)\\
\times&\det\left(\begin{array}{cc}
         y^9-x^{9} & z^9-x^{9}\\
         (y^9-1)+(y-1)^8-(x^{9}-1)-(x-1)^{8}&
         (z^9-1)+(z-1)^8-(x^{9}-1)-(x-1)^{8}
      \end{array}
\right)\\
=&(xyz)^{-5}(x-1)(y-1)(z-1)\\
\times&\det\left(\begin{array}{cc}
         y^9-x^{9} & z^9-x^{9}\\
         (y^9-1-x^{9}+1)+(y-1)^8-(x-1)^{8}&
         (z^9-1-x^{9}+1)+(z-1)^8-(x-1)^{8}
      \end{array}
\right)\\
=&(xyz)^{-5}(x-1)(y-1)(z-1)\\
\times&\det\left(\begin{array}{cc}
y^9-x^{9} & z^9-x^{9}\\
y^9-x^{9}+(y-1)^8-(x-1)^{8}
&z^9-x^{9}+(z-1)^8-(x-1)^{8}
\end{array}
\right)\\
=&(xyz)^{-5}(x-1)(y-1)(z-1)\\
\times&\det\left(\begin{array}{cc}
y^9-x^{9} & z^9-x^{9}\\
(y-1)^8-(x-1)^{8}& (z-1)^8-(x-1)^{8}
\end{array}
\right).
\end{align*}
Since the characteristics of the finite field
$\mathbb{F}_{q^2}$ is 3, one has
$$Y^9-X^{9}=(y-1)^{9}-(x-1)^9=y^9-x^9$$
and
$$Z^9-X^{9}=(z-1)^{9}-(x-1)^9=z^9-x^9.$$
This yields that
\begin{align*}
\det(A)&=XYZ((X+1)(Y+1)(Z+1))^{-5}
\det\left(\begin{array}{cc}
Y^9-X^{9} & Z^9-X^{9}\\
Y^8-X^{8}& Z^{8}-X^{8}
\end{array}\right)
\end{align*}
as required.

Since $x,y,z\in U_{q+1}\backslash\{1\}$ are
three pairwise distinct elements, we have
$$XYZ((X+1)(Y+1)(Z+1))^{-5}\ne 0.$$
It then follows immediately that $\det(A)=0$
if and only if
\begin{equation*}
\det\left(\begin{array}{cc}
        Y^{9}-X^{9} & Z^{9}-X^{9} \\
        Y^{8}-X^{8} & Z^{8}-X^{8}
      \end{array}
\right)=0
\end{equation*}
as desired.

This finishes the proof of Lemma \ref{lem11}.
\end{proof}

\begin{lem}\label{lem16}
Let $q=3^{m}$ with $m$ being an odd integer.
Let $x,y,z\in U_{q+1}\backslash\{1\}$ be
three pairwise distinct elements.
Write $X:=x-1$, $Y:=y-1$ and $Z:=z-1$.
Then it holds in $\mathbb{F}_{q^2}$ that
\begin{equation}\label{3.6}
\det\left(\begin{array}{cc}
Y^{9}-X^{9} & Z^{9}-X^{9} \\
Y^{8}-X^{8} & Z^{8}-X^{8}
\end{array}\right)=0
\end{equation}
if and only if it holds in $\mathbb{F}_{q^2}$ that
\begin{equation}\label{3.7}
\frac{(y-1)(z-x)}{(x-1)(z-y)}=-1.
\end{equation}
\end{lem}

\begin{proof}
Firstly, we let (\ref{3.6}) hold in $\mathbb{F}_{q^2}$. Then
\begin{align*}
& Y^{9}(Z^{8}-X^{8})-Y^{8}(Z^{9}-X^{9})+X^{8}Z^{9}-X^{9}Z^{8}\\
=&(Z-X)\Big(Y^{9}\frac{Z^{8}-X^{8}}{Z-X}-Y^{8}
(Z-X)^{8}+X^{8}Z^{8}\Big)=0.
\end{align*}
Since
$$Z^8-X^8=\frac{1}{Z}(Z^{9}-X^{9}-X^{8}(Z-X))
=\frac{Z-X}{Z}((Z-X)^{8}-X^{8}),$$
one has
$$(Z-X)\Big(\frac{Y^9}{Z}((Z-X)^{8}-X^{8})
-Y^{8}(Z-X)^{8}+X^8Z^8\Big)=0.$$
But $X\ne Z$ since $x\ne z$. Hence
$$\frac{Y^{9}}{Z}((Z-X)^{8}-X^{8})-Y^{8}
(Z-X)^{8}=-X^8Z^8.$$
It follows that
$$\frac{Y^{9}}{Z}(Z-X)^{8}-Y^{8}(Z-X)^{8}
=\frac{Y^{9}X^{8}}{Z}-X^8Z^8$$
from which one can deduce that
$$ Y^{8}(Z-X)^{8}\Big(\frac{Y}{Z}-1\Big)
=\frac{X^8}{Z}(Y^9-Z^9)=\frac{X^8}{Z}(Y-Z)^{9}.$$
That is,
$$ Y^{8}(Z-X)^{8}(Y-Z)=X^8(Y-Z)^{9}.$$
Since $Y-Z\ne 0$, one derives that
$$\Big(\frac{Y(Z-X)}{X(Z-Y)}\Big)^{8}=1.$$

Now we let (\ref{3.7}) be true. Let
\begin{align}\label{3.8}
\Delta:=\frac{(y-1)(z-x)}{(x-1)(z-y)}.
\end{align}
Then $\frac{Y(Z-X)}{X(Z-Y)}=\Delta\in
\mathbb{F}_{q^2}$ and $\Delta^{8}=1$.
We claim that
$\Delta\in \mathbb{F}_{q}$. Actually,
since $x,y,z\in U_{q+1}\backslash\{1\}
\subseteq\mathbb{F}_{q^2}$, we have
\begin{align*}
\Delta^{q}=& \frac{(y-1)^{q}(z-x)^{q}}{(x-1)^{q}(z-y)^{q}}\\
=& \frac{(y^{q}-1)(z^{q}-x^q)}{(x^q-1)(z^{q}-y^{q})}\\
=& \frac{(y^{-1}-1)(z^{-1}-x^{-1})}{(x^{-1}-1)(z^{-1}-y^{-1})}\\
=& \frac{x^{-1}y^{-1}z^{-1}(1-y)(x-z)}{x^{-1}y^{-1}z^{-1}(1-x)(y-z)}\\
=& \frac{(y-1)(z-x)}{(x-1)(z-y)}=\Delta.
\end{align*}
Hence $\Delta\in \mathbb{F}_{q}$ as claimed. The claim
is proved. Evidently, one has $\Delta\ne 0$. Then the
claim tells us that $\Delta^{q-1}=1$. However, since
$q=3^m$ and $m\ge 1$ is an odd integer, we have
$q-1\equiv 3-1=2\pmod 8$. Then $\gcd(q-1, 8)=2$.
So from $\Delta^{8}=1$ and $\Delta^{q-1}=1$
we deduce that $\Delta^{2}=1$. It follows that
$\Delta=1$ or $\Delta=-1$. If $\Delta(z)=1$, then
$(y-1)(z-x)=(x-1)(z-y)$. This implies that
$x-y=xz-yz$. Since $x\ne y$, it then follows that
$z=\frac{x-y}{x-y}=1$ which contradicts with the
hypothesis $z\ne 1$. Hence we must have $\Delta=-1$
as (\ref{3.7}) wanted. The necessity part is proved.

Conversely, let (\ref{3.7}) hold in $\mathbb{F}_{q^2}$. Then
$$\Big(\frac{Y(Z-X)}{X(Z-Y)}\Big)^{8}
=\Big(\frac{(y-1)(z-x)}{(x-1)(z-y)}\Big)^{8}=1.$$
Then
$$
(Z-Y)Y^8(Z-X)^9=(Z-X)X^8(Z-Y)^9.
$$
Since
$$(Z-X)^9=Z^9-X^9 \ {\rm and} \ (Z-Y)^9=Z^9-Y^9,$$
it follows that
$$
(Z-Y)Y^8(Z^9-X^9)=(Z-X)X^8(Z^9-Y^9).
$$
It reduces to
$$
Y^8Z^{10}-Y^9Z^9-X^9Y^8Z=X^8Z^{10}-X^9 Z^9-X^8Y^9 Z.
$$
Then dividing by $Z$ on both sides gives that
$$Y^{9}(Z^{8}-X^{8})-Y^{8}(Z^{9}-X^{9})
+X^{8}Z^{9}-X^{9}Z^{8}=0.$$
We then deduce that
\begin{align*}
\det\left(\begin{array}{cc}
Y^{9}-X^{9} & Z^{9}-X^{9} \\
Y^{8}-X^{8} & Z^{8}-X^{8}
\end{array}\right) &=0.
\end{align*}
Namely, (\ref{3.6}) is true. The sufficiency part is proved.

The proof of Lemma \ref{lem16} is complete.
\end{proof}

\begin{lem}\label{lem17}
Let $q=3^{m}$ with $m$ being a positive integer.
Let $x,y\in U_{q+1}$ be two elements. Then each
of the following holds in $\mathbb{F}_{q^2}$:

{\rm (i).} We have $x+y+xy=0\Longleftrightarrow x+y+1=0
\Longleftrightarrow x=y=1$.

{\rm (ii).} If either $x,y\in U_{q+1}\setminus\{1\}$,
or $x,y\in U_{q+1}$ and $x\ne y$, then $x+y+xy\ne 0$
and $x+y+1\ne 0$.
\end{lem}

\begin{proof}
(i). First of all, we show that
$x+y+xy=0\Longleftrightarrow x+y+1=0$.

In fact, since $x,y\in U_{q+1}$, we have
$x^q=x^{-1}$ and $y^q=y^{-1}$. Hence
$$x+y+xy=xy(x^{-1}+y^{-1}+1)=xy(x^q+y^q+1)=xy(x+y+1)^q.$$
But $x\ne 0$ and $y\ne 0$. It then follows that $x+y+xy=0$
holds if and only if $x+y+1=0$ holds as required.
So the first statement of part (i) is true.

Consequently, we show that $x+y+1=0\Longleftrightarrow x=y=1$.
Let $x=y=1$. Since the characteristics of ${\mathbb{F}}_{q^2}$
is equal to 3, one has $x+y+1=0$ holds in ${\mathbb{F}}_{q^2}$.
Conversely, let $x+y+1=0$. Then we have $x+y=-1=2$.
By the first part of (i), one knows that $x+y+xy=0$.
We then derive that $xy=1$ holds in ${\mathbb{F}}_{q^2}$.
It follows that $x$ and $y$ are the two roots in ${\mathbb{F}}_{q^2}$
of the quadratic equation
$$u^2-2u+1=(u-1)^2=0.$$
Thus we must have $x=y=1$ as desired. Thus part (i) is proved.

(ii). Let either $x,y\in U_{q+1}\setminus\{1\}$,
or $x,y\in U_{q+1}$ and $x\ne y$. Then $(x, y)
\ne (1, 1)$. So by part (i), we have
$x+y+xy\ne 0$ and $x+y+1\ne 0$ as desired.

This concludes the proof of Lemma \ref{lem17}.
\end{proof}

\begin{lem}\label{lem18}
Let $q=3^{m}$ with $m$ being a positive
integer. Let $x,y\in U_{q+1}\setminus\{1\}$
be two distinct elements. Then we have
$$-\frac{x+y+xy}{x+y+1}\in U_{q+1}\setminus\{x,y,1\}.$$
\end{lem}
\begin{proof}
At first, we show that
$$-\frac{x+y+xy}{x+y+1}\in U_{q+1}.$$
Since $x,y\in U_{q+1}\setminus\{1\}$, by Lemma 3.6 (ii)
we have $x+y+xy\ne 0$ and $x+y+1\ne 0$. Thus
\begin{align*}
\Big(-\frac{x+y+xy}{x+y+1}\Big)^q=&-\frac{(x+y+xy)^q}{(x+y+1)^q}\\
=&-\frac{x^{q}+y^{q}+(xy)^{q}}{x^{q}+y^{q}+1}\\
=&-\frac{x^{-1}+y^{-1}+(xy)^{-1}}{x^{-1}+y^{-1}+1}\\
=&-\frac{x+y+1}{x+y+xy}.
\end{align*}
It then follows that
$$\Big(-\frac{x+y+xy}{x+y+1}\Big)^{q+1}=1$$
which implies that $-\frac{x+y+xy}{x+y+1}\in U_{q+1}$
as one wants.

It remains to show that
$$-\frac{x+y+xy}{x+y+1}\not\in\{x,y,1\}.$$
This will be done in what follows.

First, suppose that $-\frac{x+y+xy}{x+y+1}=x$. Then
$x+y+xy+x(x+y+1)=0$. So $2x+2xy+x^2+y=0$ from which
one can derive that $-x-xy+x^2+y=0$. That is, $(y-x)(1-x)=0$.
Hence we have $x=y$ or $x=1$. This contradicts with the
hypothesis $x\ne y$ and $x\ne 1$. Therefore we must
have
$$-\frac{x+y+xy}{x+y+1}\ne x$$
as required.

Consequently, suppose that $-\frac{x+y+xy}{x+y+1}=y$.
Likewise, we can deduce that $(x-y)(1-y)=0$.
Thus either $x=y$ or $y=1$. We arrive at
a contradiction with the hypothesis $x\ne y$
and $y\ne 1$. Hence we must have
$$-\frac{x+y+xy}{x+y+1}\ne y$$
as expected.

Finally, suppose that $-\frac{x+y+xy}{x+y+1}=1$.
Then $x+y+xy+x+y+1=0$. Namely, $-x-y+xy+1=0$.
Equivalently, $(1-x)(1-y)=0$. This is impossible
since the hypothesis that $x\ne 1$ and $y\ne 1$
implies that $(1-x)(1-y)\ne 0$. In conclusion, we
have
$$-\frac{x+y+xy}{x+y+1}\ne 1$$
as desired.

This completes the proof of Lemma \ref{lem18}.
\end{proof}

\begin{lem}\label{lem19}
Let $q=3^{m}$ with $m$ being an odd integer.
Let $x,y\in U_{q+1}\setminus\{1\}$ be two distinct
elements. Then there is a unique element
$z\in U_{q+1}\setminus\{x,y,1\}$ such that the
system (\ref{3.1}) of homogeneous equations
has a nonzero solution ${\bf X}$ in
${\mathbb{F}}^{4}_{q}$.
\end{lem}
\begin{proof}
This lemma follows immediately from Lemmas
\ref{lem10} to \ref{lem18}.
\end{proof}

We can now give the proof of Theorem \ref{thm1.2}
as the conclusion of this paper.\\

\noindent {\it Proof of Theorem \ref{thm1.2}.}
Let $\mathcal{C}$ be the BCH code $\mathcal{C}_{(q, q+1, 3, 4)}$
with parameters $[q+1, q-3, 4]$. On the one hand, Theorem 1.1
tells us that $\mathcal{C}$ is an AMDS code. And Theorem
\ref{thm2.1} infers that the dual code of $\mathcal{C}$ is
an AMDS code if and only if $\dim\mathcal{C}(I)=1$ holds
for any subset $I\subseteq\langle n\rangle$ with
$|I|=n-k+1=5$. But by Lemma \ref{lem3.1}, we know that
$\dim\mathcal{C}(I)=1$ holds for any subset
$I\subseteq\langle n\rangle$ with $|I|=n-k+1=5$ if and
only if for arbitrary two distinct elements
$x, y\in U_{q+1}\backslash\{1:=1_{\mathbb{F}_{q^{2}}}\}
\subseteq \mathbb{F}_{q^{2}}$, there exists at most one element
$z\in U_{q+1}\backslash\{x, y, 1\}\subseteq{\mathbb{F}_{q^{2}}}$
such that the system (\ref{3.1}) of homogeneous equations
has a nonzero solution ${\bf X}$ in ${\mathbb{F}}^{4}_{q}$.
But the latter one is true because Lemma \ref{lem19} guarantees
that for arbitrary two distinct elements
$x, y\in U_{q+1}\backslash\{1:=1_{\mathbb{F}_{q^{2}}}\}
\subseteq \mathbb{F}_{q^{2}}$, there is a unique element
$z\in U_{q+1}\backslash\{x, y, 1\}\subseteq{\mathbb{F}_{q^{2}}}$
such that the system (\ref{3.1}) of homogeneous equations
has a nonzero solution ${\bf X}$ in ${\mathbb{F}}^{4}_{q}$.
Then we can conclude that the dual code of $\mathcal{C}$ is
an AMDS code.

This finishes the proof of Theorem \ref{thm1.2}. \hfill$\Box$

\bibliographystyle{amsplain}

\begin{thebibliography}{}
\bibitem{[Ap]} T.M. Apostol, {\it Introduction to
analytic number theory}, Springer-Verlag, New York, 1976.

\bibitem{[BGP-IT15]} D. Bartoli, M. Giulietti and I.
Platoni, On the covering radius of MDS codes,
{\it IEEE Trans. Inform. Theory} {\bf 61} (2015), 801-811.

\bibitem{[CS-HEP05]} L.S. Chen and S.Y. Shen,
{\it Fundamentals of coding theory (Chinese)},
Higher Education Press, Beijing, 2005.

\bibitem{[DB-DCC96]} M.A. de Boer, Almost MDS codes,
{\it Des. Codes Cryptogr.} {\bf 9} (1996), 143-155.

\bibitem{[DFZ-FFA17]} C.S. Ding, C.L. Fan and Z.C. Zhou,
The dimension and minimum distance of two classes of primitive
BCH codes, {\it Finite Fields Appl.} {\bf 45} (2017), 237-263.

\bibitem{[DT-IT20]} C.S. Ding and C.M. Tang, Infinite
families of near MDS codes holding $t$-designs, {\it IEEE
Trans. Inform. Theory} {\bf 66} (2020), 5419-5428.

\bibitem{[SI-JG95]} S.M. Dodunekov and I.N. Landgev, On near
MDS codes, {\it J. Geometry} {\bf 54} (1995), 30-43.

\bibitem{[DL-DM00]} S.M. Dodunekov and I.N. Landjev, Near-MDS
codes over some small fields, {\it Discrete Math.} {\bf 213}
(2000), 55-65.

\bibitem{[DDK-IT97]} R. Dodunekova, S.M. Dodunekov and
T. Klove, Almost-MDS and near-MDS codes for error detection,
{\it IEEE Trans. Inform. Theory} {\bf 43} (1997), 285-290. 

\bibitem{[GYZZ-FFA22]} X.J. Geng, M. Yang, J. Zhang
and Z.C. Zhou, A class of almost MDS codes,
{\it Finite Fields Appl.} {\bf 79} (2022), \#101996.

\bibitem{[G-AAECC15]} M. Giulietti, On the extendibility
of near-MDS elliptic codes, {\it Appl. Algebra Engrg.
Comm. Comput.} {\bf 15} (2004), 1-11.

\bibitem{[GDL-IT22]} B.K. Gong, C.S. Ding and C.J. Li,
The dual codes of several classes of BCH codes,
{\it IEEE Trans. Inform. Theory} {\bf 68} (2022), 953-964.

\bibitem{[HLW-IT22]} Z.L. Heng, C.J. Li and X.R. Wang,
Constructions of MDS, near MDS and almost MDS codes from
cyclic subgroups of ${\mathbb F}^*_{q^2}$, {\it IEEE Trans.
Inform. Theory} {\bf 68} (2022), 7817-7831.

\bibitem{[HP-CUP03]} W.C. Huffman and V. Pless,
{\it Fundamentals of error-correcting codes},
Cambridge University Press, Cambridge, 2003.

\bibitem{[JK-IT19} L.F. Jin and H.B. Kan, Self-dual near MDS
codes from elliptic curves, {\it IEEE Trans. Inform. Theory}
{\bf 65} (2019), 2166-2170.

\bibitem{[LDXG-IT17]} S.X. Li, C.S. Ding, M.S. Xiong and
G.N. Ge, Narrow-sense BCH codes over ${\rm GF}(q)$ with
length $n=\frac{q^m-1}{q-1}$, {\it IEEE Trans. Inform. Theory}
{\bf 63} (2017), 7219-7236.

\bibitem{[LH-FFA23]} X.R. Li and Z.L. Heng, Constructions
of near MDS codes which are optimal locally recoverable codes,
{\it Finite Fields Appl.} {\bf 88} (2023), \#102184.

\bibitem{[LWL-IT19]} C.J. Li, P. Wu and F.M. Liu, On two classes
of primitive BCH codes and some related codes, {\it IEEE Trans.
Inform. Theory} {\bf 65} (2019), 3830-3840.

\bibitem{[LN-CUP97]} R. Lidl and H. Niederreiter,
{\it Finite fields}, Second edition, Encyclopedia Math.
Appl. 20, Cambridge University Press, Cambridge, 1997.

\bibitem{[LLGS-DM23]} Y. Liu, R.H. Li, L.B. Guo and H. Song,
Dimensions of nonbinary antiprimitive BCH codes and some
conjectures, {\it Discrete Math.} {\bf 346} (2023), \#113496.

\bibitem{[MS-1977]} F.J. MacWilliams and N.J.A. Sloane,
{\it The theory of error-correcting codes}, North-Holland 
Publishing Company, Amsterdam, 1977.

\bibitem{[AMM-FFA22]} A. Meneghetti, M. Pellegrini and M. Sala,
A formula on the weight distribution of linear codes with
applications to AMDS codes, {\it Finite Fields Appl.} {\bf 77}
(2022), \#101933.

\bibitem{[RMR-CUP06]} R.M. Roth, {\it Introduction to
coding theory}, Cambridge University Press, Cambridge, 2006.

\bibitem{[RS-IT64]} R. Singleton, Maximum distance $q$-ary codes,
{\it IEEE Trans. Inform. Theory} {\bf 10} (1964), 116-118.

\bibitem{[WH-DM21]} Q.Y. Wang and Z.L. Heng, Near MDS codes from
oval polynomials, {\it Discrete Math.} {\bf 344} (2021), \#112277.
\end{thebibliography}

\end{document}